\documentclass[12pt,a4paper]{article}
\usepackage[latin1]{inputenc}
\usepackage{amsmath}
\usepackage{amsthm}
\usepackage{amsfonts}
\usepackage{vmargin}
\usepackage{amssymb}
\author{Christopher Shirley\thanks{The author thanks his supervisor Frédéric Klopp, for his advice and guidance over the course of the study}}
\title{Decorrelation estimates for random discrete schr\"{o}dinger operators in dimension one and applications to spectral statistics}

\newtheorem{theo}{Theorem}[section]

\newtheorem{prop}[theo]{Proposition}
\newtheorem{lem}[theo]{Lemma}

\newtheorem{defi}[theo]{Definition}
\newtheorem{rem}[theo]{Remark}
\numberwithin{equation}{section}
 
\newcommand{\R}{\mathbb{R}}
\newcommand{\N}{\mathbb{N}}
\newcommand{\Z}{\mathbb{Z}}
\def\restriction#1#2{\mathchoice
              {\setbox1\hbox{${\displaystyle #1}_{\scriptstyle #2}$}
              \restrictionaux{#1}{#2}}
              {\setbox1\hbox{${\textstyle #1}_{\scriptstyle #2}$}
              \restrictionaux{#1}{#2}}
              {\setbox1\hbox{${\scriptstyle #1}_{\scriptscriptstyle #2}$}
              \restrictionaux{#1}{#2}}
              {\setbox1\hbox{${\scriptscriptstyle #1}_{\scriptscriptstyle #2}$}
              \restrictionaux{#1}{#2}}}
\def\restrictionaux#1#2{{#1\,\smash{\vrule height .8\ht1 depth .85\dp1}}_{\,#2}}

\begin{document}
\maketitle

\begin{abstract}
The purpose of the present work is to establish decorrelation estimates for some random discrete Schrödinger operator in dimension one. We prove that the Minami estimates are consequences of the Wegner estimates and Localization. We also prove decorrelation estimates at distinct energies for the random hopping model and Schrödinger operators with alloy-type potentials. These results are used to give a description of the spectral statistics.
\end{abstract}

\section{Introduction}
To  introduce our results, let us first consider one of the random operators that will be studied in the rest of this article.
Let $(a_\omega(n))_{n\in\Z}$ be a sequence of independent, real random variables uniformly distributed on $[1,2]$. Define $H_\omega:\ell^2(\Z)\rightarrow\ell^2(\Z)$ as the following hopping model :
\begin{equation}\label{fullophopp}
\forall u\in\ell^2(\Z), [H_\omega(u)](n):=a_{\omega}(n+1) u(n+1)+ a_\omega(n)u(n-1).
\end{equation} 
We know that, with probability one, $H_\omega$ is bounded and is self-adjoint. As $H_\omega$ is $\Z$-ergodic, we know that there exists a set $\Sigma$ such that, with probability one, the spectrum of $H_\omega$ is equal to $\Sigma$ (see for instance \cite{Kirsh-invitation}). Using   \cite[Theorem 3]{k-mar}, we know that
\begin{equation}
\Sigma=[-4,4].
\end{equation}
We also know (\cite{Kirsh-invitation}) there exists a bounded density a state $E\rightarrow\nu(E)$ such that for any continuous function $\psi:\R\rightarrow\R$
\begin{equation}
\mathbb{E}\left(\left \langle\delta_0,\psi\left(H_\omega\right)\delta_0\right\rangle\right)=\int_\R \psi(E)\nu(E) dE.
\end{equation}
Let $N(E)$ be the integrated density of state, i.e the distribution function of the measure $\nu(E) dE$. 

 One of the purposes of this article is to give a description of the spectral statistics of $H_\omega$. In this context, we study the restriction of $H_\omega$ to a finite box and study the diverse statistics when the size of the box tends to infinity. For $L\in\N$, let $\Lambda_L=[1,L]\cap\Z$ and $H_\omega(\Lambda_L)$ be the restriction of $H_\omega$ to $\ell^2(\Lambda_L)$ with Dirichlet boundary conditions. In order to study the spectral statistics of $H_\omega(\Lambda)$ we use four results, the localization assumption, the Wegner estimates, the Minami estimates and the decorrelation estimates for distinct energies. They will be introduced in the rest of the section.

Let $\mathcal{I}$ be a relatively compact open subset of $\R^*$. We know from the \cite[Section IV]{klopp:4975} that the following Wegner estimates hold at the edges of the spectrum.

\textbf{(W) :} There exists $C>0$,   such that for $J\subset \mathcal{I}$ and $L\in\N$ 
\begin{equation}
\mathbb{P}\Big[\text{tr} \left(\textbf{1}_J(H_\omega(\Lambda_L)) \right)\geq 1\Big ]\leq C |J||\Lambda_L|.
\end{equation}

This shows that the integrated density of state (abbreviated IDS from now on) $N(.)$ is Lipschitz continuous. As the IDS is a non-decreasing function, this imply that $N$ is almost everywhere differentiable and its derivative $\nu(.)$  is positive almost-everywhere on its essential support.

Let $(E_j)_{j\in\{1,\dots,L\}}$ denote the eigenvalues, ordered increasingly and repeated according to multiplicity. 
One purpose of this article is, as in \cite{2010arXiv1011.1832G}, to give a description of spectral statistics. For instance, we obtain the following result. Define the \textit{unfolded local level statistics} near $E_0$ as the following point process :
\begin{equation}
\Xi(\xi;E_0,\omega,\Lambda)=\sum_{j\geq1} \delta_{\xi_j(E_0,\omega,\Lambda)}(\xi)
\end{equation}
 where
 \begin{equation}
 \xi_j(E_0,\omega,\Lambda)=|\Lambda|(N(E_j(\omega,\Lambda)-N(E_0)).
 \end{equation}
The unfolded local level statistics are described by the following theorem which corresponds to \cite[Theorem 1.9]{2010arXiv1011.1832G} with a stronger hyp

\begin{theo}\label{ULLS}
Pick $E_0\in \mathcal{I}$ such that $N(.)$ is differentiable at $E_0$ and $\nu(E_0)>0$.Then, when $|\Lambda|\to \infty$, the point process
$\Xi(\xi;E_0,\omega,\Lambda)$ converges weakly to a Poisson process with intensity the Lebesgue measure. That is, for any $p\in\N^*$, for any $(I_i)_{i\in\{1,\dots,p\}}$ collection of disjoint intervals
\begin{equation}
\lim_{|\Lambda|\to\infty}\mathbb{P}
\left(
\left\{\omega;
\begin{aligned}
\sharp\{j;\xi_j(\omega,\Lambda)\in I_1\}=k_1\\
\vdots\hspace*{8em} \vdots\hspace*{1em}\\
\sharp\{j;\xi_j(\omega,\Lambda)\in I_p\}=k_p
\end{aligned}
\right\}\right)=\dfrac{|I_1|^{k_1}}{k_1!}\dots\dfrac{|I_p|^{k_p}}{k_p!}
\end{equation}

\end{theo}

Now, one can wonder what is the joint behaviour at large scale of the point processes $\Xi(\xi;E_0,\omega,\Lambda)$ and $\Xi(\xi;E_1,\omega,\Lambda)$ with $E_0\neq E_1$. We obtain the following theorem which corresponds to \cite[Theorem 1.11]{2010arXiv1011.1832G}.

\begin{theo}\label{joint}
Pick $(E_0,E_0')\in \mathcal{I}^2$ such that $E_0\neq E_0'$ and such that $N(.)$ is differentiable at $E_0$ and $E_0'$ with $\nu(E_0)>0$ and $\nu(E_0')>0$.\\
When $|\Lambda|\rightarrow \infty$ the point processes $ \Xi(E_0,\omega,\Lambda)$ and $\Xi(E_0',\omega,\Lambda)$, converge weakly respectively to two independent Poisson processes on $\R$ with intensity the Lebesgue measure. That is, for $U_+\subset\R$ and $U_-\subset\R$ compact intervals and $\{k_+,k_-\}\in \N\times \N$, one has 
\begin{displaymath}
\mathbb{P}\left(
\begin{aligned}
\sharp\{j;\xi_j(E_0,\omega,\Lambda)\in U_+\}=k_+\\
\sharp\{j;\xi_j(E_0',\omega,\Lambda)\in U_-\}=k_-
\end{aligned}
\right)\underset{\Lambda\to\Z}{\rightarrow} \left(\dfrac{|U_+|^{k_+}}{k_+!}e^{-|U_+|}\right)\left(\dfrac{|U_-|^{k_-}}{k_-!}e^{-|U_-|}\right).
\end{displaymath}
\end{theo}

To prove these theorems we use three results, the localization assumption, the Minami estimates and the decorrelation estimates. They will be introduced in the rest of the section.

We know from \cite{klopp:4975} that the operator satisfies the following localization assumption.\\
\textbf{(Loc):} for all $\xi\in(0,1)$, one has
\begin{equation}
\sup_{L>0} 
\underset{|f|\leq 1}{\underset{\text{supp }f\subset \mathcal{I}}\sup}
\mathbb{E}\left(\sum_{\gamma\in\Z^d}e^{|\gamma|^\xi}\|\textbf{1}_{\Lambda(0)}f(H_\omega(\Lambda_L))\textbf{1}_{\Lambda(\gamma)}\|_2\right)<\infty.
\end{equation}

Now, we introduce the Minami estimates. We prove the 

\begin{theo}[M]\label{mina-D3}
For any $s'\in(0,s)$, $M>1$, $\eta>1$, $\rho\in(0,1)$, there exists $L_{s',M,\eta,\rho}>0$ and $C=C_{s',M,\eta,\rho}>0$ such that, for $E\in J,L\geq L_{s',M,\eta,\rho}$ and $\epsilon\in[L^{-1/s'}/M,ML^{-1/s'}]$ , one has
\begin{displaymath}
\sum_{k\geq2}\mathbb{P}\big(\textup{tr}[\textbf{1}_{[E-\epsilon,E+\epsilon]}(H_\omega(\Lambda_L))]\geq k\big)\leq C( \epsilon^s L )^{1+\rho}.
\end{displaymath}
\end{theo}
It is proven in \cite{2011arXiv1101.0900K} that, in dimension one, for the continuum model, if one has independence at a distance and localization, the Minami estimates are an implication of the Wegner estimates. We show that this statement also holds for discrete models, such as the random hopping model. In both cases, the Minami estimates are not as precise as the Minami estimates proven in \cite{CGA09}, but are sufficient for our purpose. (Loc) and (M) are sufficient to prove Theorem~\ref{ULLS}.

We now introduce the decorrelation estimates at distinct energies. We prove the
\begin{theo}\label{decohopp}
There exists $\gamma>0$ such that for any $\beta\in(1/2,1)$, $\alpha\in (0,1) $ and $(E,E')\in(\R^*)^2$ such that at $|E|\neq |E'|$, for any $k>0$ there exists $C>0$ such that for $L$ sufficiently large and $kL^\alpha\leq l\leq L^\alpha/k$ we have 
\begin{displaymath}
\mathbb{P}\left(
\begin{aligned}
 \textup{tr}\, \textbf{1}_{[E-L^{-1},E+L^{-1}]}\left(H_\omega(\Lambda_l)\right)\neq 0,\\
 \textup{tr}\, \textbf{1}_{[E'-L^{-1},E'+L^{-1}]}\left(H_\omega(\Lambda_l)\right)\neq 0
\end{aligned}
\right)\leq C\dfrac{l^2}{L^{1+\gamma}}.
\end{displaymath}
\end{theo}

Decorrelation estimates give more precise results about spectral statistics. They are a consequence of Minami estimates and localization. In \cite{K10}, Klopp proves decorrelation estimates for eigenvalues of the discrete Anderson model in the localized regime. The result is proven at all energies only in dimension one. In \cite{phong12}, decorrelation estimates are proven for the one-dimensional tight binding model, i.e when there are correlated diagonal and off-diagonal disorders. As we used (M) in the proof of Theorem~\ref{ULLS} we use Theorem~\ref{decohopp} to prove Theorem~\ref{joint}.
\section{Models and main results}
Let $(a_\omega(n))_{n\in\Z}\in\R^\Z$ and $(V_\omega(n))_{n\in\Z}\in\R^\Z$ be two random sequences. Define $\Delta_a:\ell^2(\Z)\rightarrow\ell^2(\Z)$ by
\begin{equation}\label{defjacmod}
\Delta_a u(n)=a_{\omega}(n+1) u(n+1)+ a_\omega(n)u(n-1)
\end{equation}
and define $H_\omega:\ell^2(\Z)\rightarrow\ell^2(\Z)$ as the following Jacobi operator :
\begin{equation}\label{fullop}
\forall u\in\ell^2(\Z), [H_\omega(u)](n):=(\Delta_a u)(n)+ V_\omega(n) u(n).
\end{equation}

For $L\in\N$, let $\Lambda_L=[1,L]\cap\Z$ and $H_\omega(\Lambda_L)$ be the restriction of $H_\omega$ to $\ell^2(\Lambda_L)$ with Dirichlet boundary conditions.
We will study these operators in two cases :\\
$(\bullet)$ $(V_\omega(n))_{n\in\Z}$ are random variables not necessarily independent with a common compactly supported bounded density $\rho$ and the sequence $(a_\omega(n))_{n\in \Z}$ is deterministic (or random provided it is independent with $V_\omega$) and there exists $M>0$ such that $\forall n,\dfrac{1}{M}\leq |a(n)|\leq M$. 

For instance, fix $(a(n))_n$ and for $i\in\Z$, fix $u_i\in\ell^1(\Z)$. We may take $V_\omega(n)=\sum_{i\in\Z}\omega_i u_i(n-i)$ where $(\omega_i)_{i\in\Z}$ are i.i.d random variables. The functions $(u_i)_i$ are called single-site potentials. When all $(u_i)_i$ are equal, the potential is said alloy-type. In particular, when $u_i=\delta_0$, we obtain the Anderson model. \\
$(\bullet)$ $(a_\omega(n))_{n\in\Z}$ are random variables not necessarily independent with a common compactly supported bounded density $\mu$ and the sequence $(V_\omega(n))_{n\in \Z}$ is deterministic or random (not necessarily independent with $a_\omega$). 

For instance, fix $V_\omega:=0$ and take $(a_\omega(n))_n$ i.i.d random variables, then we obtain the following random hopping model : 
\begin{equation}\label{horsdiag}
H_\omega(\Lambda_L)=\begin{pmatrix}
0 & a_\omega(1) & 0 & \hdots    \\
a_\omega(1) & 0 & a_\omega(2) & 0   & \hdots \\ 
0 & a_\omega(2) & 0 &a_\omega(3) & 0 & \hdots\\
 & \ddots & \ddots & \ddots & \ddots & \ddots 
 \end{pmatrix}.
 \end{equation}
\\
We assume that the operator satisfies a condition of independence at  a distance :\\
\textbf{(IAD) : }There exist $S\in\N$ such that  for all $n\in\N$ and $(\Lambda_k)_{k\in\{1,\dots,n\} }$ any collection of intervals in $\Z$, if for $k\neq k'$ $\text{dist}(\Lambda_k,\Lambda_k')\geq S$ then the operators $(H_\omega(\Lambda_k))_{k\in\{1,\dots,n\} }$ are independent.

The discrete alloy-type Schrödinger operator with compactly supported single site potential and the hopping model satisfy (IAD).
\\Let $\mathcal{I}$ be a relatively compact open subset of $\R$. Suppose a Wegner estimate holds on $\mathcal{I}$ : \\
\textbf{(W) :} There exists $C>0$, $s\in(0,1]$, $\rho\geq 1$  such that for $J\subset \mathcal{I}$ and $\Lambda\subset\N$
\begin{equation}
\mathbb{P}\Big[tr \left(\textbf{1}_J(H_\omega(\Lambda)) \right)\geq 1\Big ]\leq C |J|^s|\Lambda|^\rho.
\end{equation}
Wegner estimate has been proven for many different models, discrete or continuous (see \cite{CGA09} for the Anderson model). In this paper, we  rely on \cite{springerlink:10.1007/s00023-010-0052-5} when the potential is alloy-type, on \cite{klopp:4975} for the hopping model and on \cite{Kl95} (see also \cite{Hislop01theintegrated}) for the other models. When $\rho=1$, if the integrated density of state (which we will write IDS from now on) exists, (W) implies that the IDS is Hölder continuous whereas when $\rho>1$, (W) gives no information about the IDS. 
\\
Finally, we suppose that there exists $\mathcal{I}\subset \R$ such that we have a localization property.

\textbf{(Loc): } for all $\xi\in(0,1)$, one has
\begin{equation}
\sup_{L>0} 
\underset{|f|\leq 1}{\underset{\text{supp }f\subset \mathcal{I}}\sup}
\mathbb{E}\left(\sum_{\gamma\in\Z^d}e^{|\gamma|^\xi}\|\textbf{1}_{\Lambda(0)}f(H_\omega(\Lambda_L))\textbf{1}_{\Lambda(\gamma)}\|_2\right)<\infty
\end{equation}

This property can be shown using either multiscale analysis or fractional moment method. In fact we suppose that $\mathcal{I}$ is a region where we can do the bootstrap multiscale analysis of \cite{boot}. (Loc) is equivalent to the conclusion of the bootstrap MSA (see \cite[Appendix]{2010arXiv1011.1832G} for details). We do not require estimates on the operator $H_\omega$ but only on $H_\omega(\Lambda_L)$. In order to make the multiscale analysis one needs an initial estimate about the Green function $(H(\Lambda_{L_0})-E)^{-1}$ where $L_0$ is large enough. For the Anderson model in dimension one, this estimate is valid on the whole axis. For alloy-type potential with undefined sign, we only know it at large disorder (see \cite{alloytype} and reference therein), i.e when $\|\rho\|_\infty$ is sufficiently small. For the hopping model \eqref{horsdiag}, (Loc) is proven at the edges of the almost-sure spectrum \cite{klopp:4975}.

The purpose of this article is to give a description of the spectral statistics for some discrete random Schrödinger operators in dimension one. 
Two key tools to get this description are the Wegner estimates and the following Minami estimates.  For instance we prove

\begin{theo}\label{mina-D2}
For any $s'\in(0,s)$, $M>1$, $\eta>1$, $\rho\in(0,1)$, there exists $L_{s',M,\eta,\rho}>0$ and $C=C_{s',M,\eta,\rho}>0$ such that, for $E\in J,L\geq L_{s',M,\eta,\rho}$ and $\epsilon\in[L^{-1/s'}/M,ML^{-1/s'}]$ , one has
\begin{displaymath}
\sum_{k\geq2}\mathbb{P}\big(\textup{tr}[\textbf{1}_{[E-\epsilon,E+\epsilon]}(H_\omega(\Lambda_L))]\geq k\big)\leq C( \epsilon^s L )^{1+\rho}.
\end{displaymath}
\end{theo}
The Wegner estimates have been proven for many type of model whereas the Minami estimates are known for all dimension mostly for the Anderson model(\cite{CGA09}). Following \cite{2011arXiv1101.0900K}, we prove that the Minami estimates are consequences of localization and the Wegner estimates
 This will be applied to both the random hopping model and the discrete alloy-type model.

We prove the decorrelation estimates.
\begin{theo}[\textbf{D}]: Let $\beta\in(1/2,1)$. For any $\alpha\in(0,1)$,  $E,E'\in\mathcal{I}^2$, $E\neq E'$, for any $c>1$, one has, for $L$  sufficiently large and $l$ such that $cL^\alpha\leq l\leq c L^\alpha$ 
\begin{displaymath}
\mathbb{P}\left(
\begin{aligned}
 tr\textbf{1}_{[E-L^{-1},E+L^{-1}]}(H_\omega(\Lambda_l))\neq 0,\\
 tr\textbf{1}_{[E'-L^{-1},E'+L^{-1}]}(H_\omega(\Lambda_l)))\neq 0
\end{aligned}
\right)\leq C\left(\dfrac{l^2}{L^{4/3}}\right)e^{(\log L)^\beta}.
\end{displaymath}
\end{theo}

 We show that decorrelation estimates hold for the random hopping model and for the Jacobi model with alloy-type potential under certain requirement, such as a compactly supported, non-negative single-site potential.

The estimates (W), (M) and (D), and assumption (Loc) will be used as in \cite{2010arXiv1011.1832G} to get a description of spectral statistics.

The Wegner estimate used in \cite{2010arXiv1011.1832G} is linear in volume. As all our previous results hold provided we have (Loc) and the Wegner estimates, even if it is polynomial in volume, we modify the proof of \cite{2010arXiv1011.1832G}. In our case, the Wegner estimates do not imply that the integrated density of state is Hölder continuous. Thus, we make an assumption about its regularity. But, in dimension one, for some models, one can prove that the integrate density of state is Hölder continuous without using the Wegner estimates.
Eventually, we use our previous results to study the spectral statistics of one-dimensional quantum graphs with random vertex coupling.
  
In section 3 we will show that in dimension one, if one has a condition of independence at distance, for the Jacobi model, the Minami estimates are consequences of the Wegner estimates. In section 4 we will prove a general decorrelation estimates that will be used for alloy-type potential and one-dimensional quantum graphs. We also prove decorrelation estimates for some discrete alloy-type models and for the random hopping model. In section 5 we prove the results about spectral statistics under a regularity hypothesis on the IDS. In section 6 we show that the previous results are also true for one-dimensional quantum graphs with random vertex coupling. Eventually, in the appendix, we give some general results about finite-difference equations used in section 3 and an example of operator which has Wegner estimates polynomial in volume and a Hölder continuous integrated density of state.
\section{In the localized regime, the Wegner estimate implies the Minami estimate}\label{sec1}

	 For any interval $\Lambda$ in $\Z$ we define $H_\omega(\Lambda)$ as the restriction of $H_\omega$ to $\ell^2(\Lambda)$ with Dirichlet boundary condition. We follow the proof of \cite{2011arXiv1101.0900K} and show that if $H_\omega(\Lambda)$ has two close eigenvalues then there exists two restrictions of this operator to disjoint intervals, each having an eigenvalue close to the previous ones. In this section, the random variables $a_\omega$ and $V_\omega$ may be correlated, as long as we have a Wegner estimate. 

\begin{theo}\label{inv}
Fix $S\geq 0$. There exists $\epsilon_0$ such that for all $\epsilon<\dfrac{\epsilon_0}{L^4}$, if \\
$tr \left( \textbf{1}_{[E-\epsilon,E+\epsilon]}(H_\omega(\Lambda_L))\right)\geq 2$ then there exists $x_-$ and $x_+$ in $\{1,\dots,L\}$ with $x_+-x_-\geq S$ such that 
\begin{align*}
tr\left(\textbf{1}_{\left[E-\epsilon\frac{L^4}{\epsilon_0},E+\epsilon\frac{L^4}{\epsilon_0}\right]}(H_\omega(\Lambda_{\{0,\dots,x_-\}})\right)\geq 1,\\ tr\left(\textbf{1}_{\left[E-\epsilon\frac{L^4}{\epsilon_0},E+\epsilon\frac{L^4}{\epsilon_0}\right]}(H_\omega(\Lambda_{\{x_+,\dots,L\}})\right)\geq 1.
\end{align*}
\end{theo}

This result is the discrete equivalent of Theorem 2.1 in \cite{2011arXiv1101.0900K}. One of the advantages of the discrete model over the continuous model is the following one. As we will apply this result to a random operator, the points $x_+$ and $x_-$ will be random. Thus, when estimating probabilities, we will have to count the number of possibilities. Hence, when studying the continuous model, we have to restrain the points to a lattice whereas in the discrete model, the lattice is given by the model. However, in the discrete model, some results need more computations, they will be given in Appendix A.

In order to remove the condition $\epsilon<\frac{\epsilon_0}{L^4}$ we use the localization assumption and reduce the study of the operator on boxes of size $L$ to the study on boxes of size $\log L$ to which we will apply Theorem~\ref{inv}. Indeed, the localization assumption yields two different characterizations of the localized regime, the proofs of which use (W) and can be found in \cite[Appendix]{2010arXiv1011.1832G}. 

\begin{lem}\textbf{(Loc)(I)} : For all $p>0$ and $\xi\in(0,1)$, for L sufficiently large, there exists a set of configuration $\mathcal{U}_{\Lambda_l}$ of probability larger than $1-L^{-p}$ such that if $\phi_{n,\omega}$ is a normalized eigenvector associated to the eigenvalue $E_{n,\omega}$ and $x_0(\omega)\in \{1,\dots,L\}$ maximize $|\phi_{n,\omega}|$ ($x_0$ is called a localization center) then 
\begin{equation}\label{expdec}
|\phi_{n,\omega}(x)|\leq L^{p+d} e^{-|x-x_0|^{\xi}}.
\end{equation}
\end{lem}

\begin{lem}
\textbf{(Loc)(II)} : For all $\xi\in(0,1)$ and $\nu\in(0,\xi)$, for L sufficiently large, there exists a set of configuration $\mathcal{V}_{\Lambda_l}$ of  probability larger than $1-e^{-L^{\nu}}$ such that if $\phi_{n,\omega}$ is a normalized eigenvector associated to the eigenvalue $E_{n,\omega}$ and $x_0(\omega)\in \{1,\dots,L\}$ maximize $|\phi_{n,\omega}|$ ($x_0$ is called a localization center) then 
\begin{equation}
|\phi_{n,\omega}(x)|\leq e^{2L^{\nu}} e^{-|x-x_0|^{\xi}}.
\end{equation}
\end{lem}

Using (Loc)(I) and (Loc)(II), we obtain the following theorems. Fix $J$ a compact in $\mathcal{I}$. When the energy interval $[E-\epsilon,E+\epsilon]$ is sufficiently small in comparison with the size of the box $\Lambda_L$ we prove : 
\begin{theo}\label{minasmall}
For any $s'\in(0,s)$, $p>0$ and $\eta>1$, there exists $L_{\eta,s'}>0$ such that for $E\in J,L\geq L_{\eta,s'}$ and $\epsilon\in(0,L^{-1/s'})$, one has
\begin{equation}
\sum_{k\geq 2 } \mathbb{P}(tr[\textbf{1}_{[E-\epsilon,E+\epsilon]}(H_\omega(\Lambda_L))]\geq k)\leq C(\epsilon^{2s'}L^2+ L^{-p})
\end{equation}
\end{theo}

 The following theorem is useful for proving decorrelation estimates for eigenvalues.
\begin{theo}\label{mina-D}
For any $s'\in(0,s)$, $M>1$, $\eta>1$, $\beta> 1+4s+\rho$, there exists $L_{s',M,\eta,\beta}>0$ and $C=C_{s',M,\eta,\beta}>0$ such that, for $E\in J,L\geq L_{s',M,\eta,\beta}$ and $\epsilon\in[L^{-1/s'}/M,ML^{-1/s'}]$ , one has
\begin{displaymath}
\sum_{k\geq2}\mathbb{P}\big(tr[\textbf{1}_{[E-\epsilon,E+\epsilon]}(H_\omega(\Lambda_L))]\geq k\big)\leq C( \epsilon^s L )^2(\log L)^{2\eta\beta}.
\end{displaymath}
\end{theo}

We also prove the following result which is useful if we want to take boxes of size $(\log L)^{1/\xi}$ and energy intervals polynomially small. We then use (Loc)(II) instead of (Loc)(I) :

\begin{theo}\label{mina3}
For any $s'\in(0,s)$, $M>1$, $\xi'\in(0,1)$, $\alpha\in(\xi',1)$, $\nu\in(\xi',\alpha)$, $\beta> 1+4s+\rho$, there exists $L_{s',M,\alpha,\beta,\xi',\nu}>0$ and $C=C_{s',M,\alpha,\beta,\xi',\nu}>0$  such that, for $E\in J,L\geq L_{\alpha,\beta,\xi',\nu}$ and $\epsilon\in[L^{-1/s'}/M,ML^{-1/s'}]$, one has, with $l=(\log L)^{1/\xi'}$ and $l'=l^\alpha$
\begin{displaymath}
\sum_{k\geq2}\mathbb{P}\big(tr[\textbf{1}_{[E-\epsilon,E+\epsilon]}(H_\omega(\Lambda_l))]\geq k\big)\leq C\big( \epsilon^s l \, l'^\beta\big)^2,
\end{displaymath}

\end{theo}

\begin{rem}
\normalfont If we apply directly Theorem~\ref{inv} without using (Loc), we get the following estimate :
\begin{equation}\label{sufmina}
\sum_{k\geq2}\mathbb{P}\big(tr[\textbf{1}_{[E-L^{-\kappa},E+L^{-\kappa}]}(H_\omega(\Lambda_l))]\geq k\big)\leq (L^{-s\kappa} \, l^{4s+1+\rho})^2
\end{equation}
The estimate in Corollary~\ref{mina3} is better as we have $\alpha\in(\xi',1)$, but \eqref{sufmina} could suffice in some cases.

\end{rem}

\subsection{Proof of Theorem~\ref{inv}}

Let $E\in sp(H_\omega(\Lambda_L))$ and $u$ a normalized eigenfunction. For $n\in \{1, \dots, L \}$ define $U(n):=\begin{pmatrix} u(n)\\u(n-1) \end{pmatrix}$ and the Pr\"{u}fer variables $r_u$ and $\phi_u$ by: 
\begin{equation}\label{pruf}
U(n)=r_u(n)\begin{pmatrix} \sin(\phi_u(n))\\ \cos(\phi_u(n)) \end{pmatrix}.
\end{equation}

We follow the proof of \cite[Theorem 2.1]{2011arXiv1101.0900K} and divide the proof in two part. First, we show the result when both eigenvectors live essentially far away from each other. Secondly, we show the result when both eigenvectors have a common site where theirs amplitudes $r_u$ and $r_v$ are large. We will need some properties of finite-difference equations of order two, they can be found in  Appendix A. Recall that M is fixed such that 
\begin{equation}\label{bounded}
\forall n\in\N, \dfrac{1}{M}\leq |a_n|\leq M.
\end{equation}
We first prove : 
\begin{theo}\label{part1}
There exists $\eta_0>0$ such that for all $\eta\in(0,\eta_0)$ and all $E<\dfrac{\eta^4}{L^4}$, if \\
$tr\left(\textbf{1}_{[-E,E]}( sp(H_\omega(\Lambda_L))\right)\geq 2$ and if for all $n\in\Lambda_L$, $r_u(n)r_v(n)\leq \dfrac{\eta}{L}$ where $u$ and $v$ are normalized eigenvectors, there exists $x_-$ and $x_+$ in $\Lambda_L$ with $x_+-x_-\geq S$ such that 
\begin{align*}
tr \left(\textbf{1}_{[-EL^4/\eta^4,EL^4/\eta^4]}(H_\omega(\{0,\dots,x_-\})) \right)\geq 1,\\
tr \left(\textbf{1}_{[-EL^4/\eta^4,EL^4/\eta^4]}(H_\omega(\{x_+,\dots,L\})) \right)\geq 1.
\end{align*}
\end{theo}

\begin{proof}
As $u$ and $v$ are normalized, there exists $x_u$ and $x_v$ such that $r_u(x_u)\geq \dfrac{1}{\sqrt{L}}$ and $r_v(x_v)\geq \dfrac{1}{\sqrt{L}}$. So $r_v(x_u)\leq \dfrac{\eta}{\sqrt{L}}$ and $r_u(x_v)\leq \dfrac{\eta}{\sqrt{L}}$. Without loss of generality, we can suppose that $x_u<x_v$. Let $f=\dfrac{r_u}{r_v}$ then $f(x_u)\geq \dfrac{1}{\eta}$ and $f(x_v)\leq \eta$. By Lemma~\ref{progress} there exists $C>0$ (only depending on M defined in \eqref{bounded}) such that
\begin{displaymath}
\frac{1}{C}\leq \frac{f(n+1)}{f(n)}\leq C
\end{displaymath}
So for $\eta\leq \dfrac{1}{C^{S+2}}<\dfrac{1}{C^2}$ we know that there exists $n$ such that $x_u<n<x_v$ and $\dfrac{1}{\sqrt{C}}<f(n)<\sqrt{C}$ so that $\text{min}\{n-x_u,x_v-n\}\geq \left\lfloor\frac{S}{2}\right\rfloor+2$, where $\lfloor . \rfloor$ denote the floor function. Now take $x_-=n-\left(\left\lfloor\dfrac{S}{2}\right\rfloor+2\right)$ and $x_-=n+\left(\left\lfloor\dfrac{S}{2}\right\rfloor+2\right)$, then one has $x_+-x_-\geq S+2$ and 
\begin{displaymath}
\{f(x_-),f(x_+)\}\subset \left[\frac{1}{C^{\frac{S}{2}+2}},C^{\frac{S}{2}+2}\right].
\end{displaymath}
 Using Lemma~\ref{not0} we can assume that $|\sin(\phi_v(x_-))|\geq \eta^{1/4}$ otherwise we change $x_-$ to $x_--1$. Define $\lambda=\dfrac{u(x_-)}{v(x_-)}$ and $w=u-\lambda v$. Then, $\lambda\leq \dfrac{K}{\eta^{1/4}}$, $w(0)=w(x_-)=0$ and \\
 $\|\left(H_\omega(\{0,\dots,x_-\})-E\right)w\|_{\ell^2\{0,\dots,x_-\}}\leq E$. We compute
\begin{displaymath}
\|w\|^2_{\ell^2\{0,\dots,x_-\}}\geq \|u\|^2_{\ell^2\{0,\dots,x_-\}}-2\lambda\dfrac{\eta}{L}\geq r_u(x_u)^2-K \dfrac{\eta^{3/4}}{L}\geq \dfrac{1}{2L},
\end{displaymath}
for $\eta$ sufficiently small. Thus $H_\omega(\{0,\dots,x_-\})-E$ has an eigenvalue in $E\sqrt{2L}(-1,1)$. This conclude the proof for $x_-$. The proof for $x_+$ is similar : one has to estimate $\|H_\omega(\{x_+,\dots,L\})w\|$ instead of $\|\left(H_\omega(\{0,\dots,x_-\})-E\right)w\|$.
\end{proof}
Now we prove a lemma that will be used in the second part of the proof of Theorem~\ref{inv}.

\begin{lem}\label{wrons}

For all $n\in\{0,\dots,L\}$, 
\begin{equation}
|r_u(n)r_v(n)\sin(\delta\phi(n))|\leq ME, 
\end{equation}
where M is defined in \eqref{bounded}.
\end{lem}
\begin{proof}

 Let $W(n)=a(n)\big[u(n)v(n-1)-u(n-1)v(n)\big]$ be the Wronskian of $U$ and $V$. Then, we compute
 \begin{align*}
W(n+1)=&a(n)r_u(n)r_v(n)\sin(\delta\phi(n))\\
=&a(n+1)\big[u(n+1)v(n)-u(n)v(n+1)\big]\\
=&v(n)\big(-\omega_n u(n)-a(n)u(n-1)\big)\\
&-u(n)\big(-\omega_n v(n)+Ev(n)-a(n)v(n-1)\big)\\
=&W(n)-Ev(n)u(n)
\end{align*}
We know that $W(1)=0$ and that $u$ and $v$ are normalized, so for all $n\in\{0,\dots,L\}$
\begin{displaymath}
|W(n)|=E\left|\sum_{i=0}^n u(n)v(n)\right|\leq E \sqrt{\sum_{i=0}^n u^2(n)}\sqrt{\sum_{i=0}^n v^2(n)}\leq E.
\end{displaymath}
As $|a_n|\geq\dfrac{1}{M}$, the proof of Lemma~\ref{wrons} is complete.

\end{proof}

The second part of the proof of Theorem~\ref{inv} is the following result :
\begin{theo}\label{part2}
There exists $\eta_0>0$ such that for all $\eta\in(0,\eta_0)$ and $E<\dfrac{\eta^4}{L^4}$, if  $ tr\left(\textbf{1}_{[-E,E]}( sp(H_\omega(\Lambda_L))\right)\geq 2$ and if there exists $x_0\in\{0,\dots,L\}$ such that \\
$r_u(x_0)r_v(x_0)\geq \dfrac{\eta}{L}$ where $u$ and $v$ are normalized eigenvectors, then there exists $x_-$ and $x_+$ in $\{1,\dots,L\}$ with $x_+-x_-\geq S$ such that 
\begin{align*}
tr\left(\textbf{1}_{[E-\epsilon,E+\epsilon]}(H_\omega(\{0,\dots,x_-\}))\right)\geq 1,\\
tr\left(\textbf{1}_{[E-\epsilon,E+\epsilon]}(H_\omega(\{x_+,\dots,L\}))\right)\geq 1.
\end{align*}
\end{theo}

\begin{proof}
\normalfont
Using Lemma~\ref{progress} and Lemma~\ref{not0} we can suppose that \\$|\sin(\phi_u(x_0))|\geq\dfrac{1}{C}$.
From Lemma~\ref{wrons} we know  that $|\sin(\delta\phi(x_0))|\leq \dfrac{EL}{\eta}$. Now we can suppose that $\delta\phi(x_0)\in\left[0,\dfrac{EL}{2\eta}\right]$, the case $\delta\phi(x_0)\in\left[\pi-\dfrac{EL}{2\eta},\pi\right]$ is handled in the same way. To prove Theorem~\ref{part2} we use the   

\begin{lem}\label{prufopp}
There exists $x_2\in \{1,\dots,L\}$ such that $\delta\phi(x_2)\in\left[\pi-\dfrac{EL^3}{2\eta},\pi\right]$.
\end{lem}
\begin{proof}
\normalfont
As $u$ and $v$ are orthogonal
\begin{displaymath}
\left|  r_u(x_0)r_v(x_0)\sin^2(\phi_u(x_0))+\sum_{n\neq x_0} r_u(n)r_v(n)\sin(\phi_u(n))\sin(\phi_v(n))\right|\leq C\dfrac{EL}{\eta}
\end{displaymath}
and as $EL^4\leq \eta^4$, for $L$ sufficiently large we have
\begin{equation*}
 \sum_{n\neq x_0} r_u(n)r_v(n)\sin(\phi_u(n))\sin(\phi_v(n))\leq-\dfrac{\eta}{CL}\left(1-\dfrac{C\eta^2}{L^2}\right)\leq-\dfrac{\eta}{2CL}.
\end{equation*}
Now, we have
\begin{displaymath}
\left|\sum_{n\neq x_0,r_u(n)r_v(n)\leq \eta/ L^3} r_u(n)r_v(n)\sin(\phi_u(n))\sin(\phi_v(n))\right|\leq C\dfrac{\eta}{L^2}.
\end{displaymath}

So
\begin{align*}
\sum_{n\neq x_0,r_u(n)r_v(n)\geq \eta/ L^3} r_u(n)r_v(n)\sin(\phi_u(n))\sin(\phi_v(n))&\leq -\dfrac{\eta}{2CL}\left(1-\dfrac{1}{L}\right)\\
&\leq -\dfrac{\eta}{3CL}.
\end{align*}
Now, as \, $\sin(\phi_v)=\sin(\phi_u+\delta\phi)=\sin(\phi_u)\cos(\delta\phi)+cos(\phi_u)\sin(\delta\phi)$
and as \\
$|r_ur_v\sin(\delta\phi)|\leq ME\leq \dfrac{M\eta^4}{L^4}$,  for $L$ sufficiently large we have
\begin{equation}
\sum_{n\neq x_1,r_u(n)r_v(n)\geq \eta/ L^3} r_u(n)r_v(n)\sin^2(\phi_u(n))\cos(\delta\phi_v(n))\leq -\dfrac{\eta}{4CL}.
\end{equation}
So there exists $x_2\neq x_0$ such that $r_u(x_2)r_v(x_2)\geq \dfrac{\eta}{L^3}$ and $\cos(\delta\phi(x_2))<0$. So $|\sin(\delta\phi(x_2))|\leq ME\dfrac{L^3}{\eta}$ and $\delta\phi(x_2)\in\left[\pi-M\dfrac{EL^3}{\eta},\pi\right]$. This complete the proof of Lemma~\ref{prufopp}.
\end{proof}

Now, we have $x_0$ and $x_2$ such that $\delta\phi(x_0)\in\left[0,\dfrac{EL}{\eta}\right]$ and $\delta\phi(x_2)\in\left[\pi-\dfrac{EL^3}{\eta},\pi\right]$. Thus, by Lemma~\ref{oppo} and Lemma~\ref{coli}, there exists $\epsilon_0>0$ such that for $L$ sufficiently large, there exists $x_1\in\{1,\dots,L\}$ such that $\delta\phi(x_1)\in[\epsilon_0,\pi-\epsilon_0]$. Hence if we take $x_-=x_1-\left\lfloor\dfrac{S}{2}\right\rfloor-1$ and $x_+=x_1+\left\lfloor\dfrac{S}{2}\right\rfloor+1$, then, there exists $C>0$ (depending on S) such that
\begin{equation}
\forall n\in\{x_-,\dots,x_+\},\delta\phi(n)\in\left[\dfrac{\epsilon_0}{C},\pi-\dfrac{\epsilon_0}{C}\right].
\end{equation}
We will now show that $H^-_\omega:=H_\omega(\{0,\dots,x_-\})$ has an eigenvalue in \\
$[-EL^4/\eta^4,EL^4/\eta^4]$. $H^+_\omega$ is handled in the same way.\\
First, suppose there exists $n\in\{x_-,x_-+1\}$ such that either $u(n)=0$ or $v(n)=0$. For instance, suppose that $u(x_-)=0$, then as $x_0<x_-$ and as $u^2(x_0)+u^2(x_0-1)\neq 0$, $u$ is an eigenvector of $H^-_\omega$. If $v(x_-)=0$, then $v$ is an eigenvector of $H^-_\omega$.\\
Now suppose that none of $\{u(x_-),u(x_-+1),v(x_-),v(x_-+1)\}$ is equal to zero. Without loss a generality one can suppose that $r_u(x_-)>r_v(x_-)$, if not we exchange $u$ and $v$. We compute
\begin{align*}
\left|\dfrac{u(x_-+1)}{v(x_-+1)}-\frac{u(x_-)}{v(x_-)}\right|&=\dfrac{|W(x_-+1)|}{a(x_-+1)v(x_-+1)v(x_-)}\\
&\geq\dfrac{|W(x_-+1)|}{a(x_-+1)r_v(x_-+1)r_v(x_-)}\\
&\geq\dfrac{|W(x_-+1)|}{a(x_-+1)r_v(x_-+1)r_u(x_-)}\\
&\geq\dfrac{|W(x_-+1)|}{a(x_-+1)Cr_v(x_-+1)r_u(x_-+1)}\\
&\geq\dfrac{|\sin(\delta\phi(x_-+1))|}{C}\geq\dfrac{1}{K}\geq 2\eta_0
\end{align*}
for $\eta_0$ sufficiently small. So we can suppose that $\left|\dfrac{u(x_-)}{v(x_-)}-\dfrac{r_u(x_0)}{r_v(x_0)}\right|\geq \eta_0$, if not we replace $x_-$ by $x_-+1$.\\
Let $\lambda_-=\dfrac{u(x_-)}{v(x_-)}$ and $w=u-\lambda_- v$, then $w(0)=w(x_-)=0$ and 
\begin{equation}
\|(H^-_\omega-E)(w)\|_{\ell^2(\{0,\dots,x_-\} )}=\|Ew\|_{\ell^2(\{0,\dots,x_-\})}\leq E
\end{equation}
On the other hand, we know that $r_v(x_0)>\dfrac{\eta}{L}$ so we can suppose that $v(x_0)^2\geq\dfrac{\eta^2}{2L^2}$. If this is not the case, replace $x_0$ by $x_0-1$.
\begin{align*}
w(x_0)&=r_u(x_0)\sin(\phi_u(x_0))-\lambda_-r_v(x_0)\sin(\phi_v(x_0))\\
&=[r_u(x_0)-\lambda_-r_v(x_0)]\sin(\phi_v(x_0))+ r_u(x_0)[\sin(\phi_u(x_0))-\sin(\phi_v(x_0))]\\
&=\left[\dfrac{r_u(x_0)}{r_v(x_0)}-\lambda_-\right]v(x_0)+ r_u(x_0)[\sin(\phi_u(x_0))-\sin(\phi_v(x_0))]
\end{align*}
But as u is normalized, we have
\begin{enumerate}
\item
$r_u(x_0)\leq \sqrt{2}$
\item
$|\sin(\phi_u(x_0))-\sin(\phi_v(x_0))|\leq \delta\phi(x_0)\leq \dfrac{EL}{C\eta}\leq\dfrac{\eta^3}{L^3}$
\end{enumerate}
so $\|w\|_{\ell^2(\{0,\dots,x_-\})}\geq |w(x_0)|\geq \dfrac{\eta_0\eta}{\sqrt{2}L}-\dfrac{\eta^3}{L^3}\geq\dfrac{\eta_0\eta}{2L}$ for $\dfrac{\eta}{L^2}$ sufficiently small. So  $H_\omega^-$ has an eigenvalue at distance at most $\dfrac{EL}{\eta_0\eta}$ from E. This complete the proof of Theorem~\ref{part2}.
\end{proof}

Theorem~\ref{inv} now follows from both Theorem~\ref{part1} and Theorem~\ref{part2}, taking $\epsilon_0$ sufficiently small.

\subsection{Proof of the Minami estimates}
Now we follow the proof in \cite[Section 3]{2011arXiv1101.0900K} and use Theorem~\ref{inv} to show the following Minami estimates. 
\begin{theo}\label{mina}
Fix J compact in $\mathcal{I}$ the region of localization. Then, for $p>0$, $\eta>1, \beta> \text{max}(1+4s,\rho)$ and $\rho'>\rho$ ($\rho$ and s are defined in (W)), then there exists $L_{\eta,\beta,\rho'}>0$ and $C=C_{\eta,\beta,\rho'}>0$ such that, for $E\in J,L\geq L_{\eta,\beta,\rho'}$ and $\epsilon\in(0,1)$, one has
\begin{displaymath}
\sum_{k\geq2}\mathbb{P}\big(tr[\textbf{1}_{[E-\epsilon,E+\epsilon]}(H_\omega(\Lambda_L))]\geq k\big)\leq C\Big[\big( \epsilon^s L l^\beta+e^{-l/8}\big)^2e^{C\epsilon^s L l^{\rho'}}+L^{-p}\Big],
\end{displaymath}
where $l:=(\log L)^\eta$.

\end{theo}

\begin{theo}\label{mina2}
Fix J compact in $\mathcal{I}$ the region of localization. Then, for $\alpha\in(0,1)$, $\nu\in(0,\alpha)$, $\beta> 1+4s+\rho$ and $\rho'>\rho$ ($\rho$ and s are defined in (W)), there exists $L_{\alpha,\beta,\rho'}>0$ and $C=C_{\alpha,\beta,\rho'}>0$  such that, for $E\in J,L\geq L_{\alpha,\beta,\rho'}$ and $\epsilon\in(0,1)$, one has
\begin{displaymath}
\sum_{k\geq2}\mathbb{P}\big(tr[\textbf{1}_{[E-\epsilon,E+\epsilon]}(H_\omega(\Lambda_L))]\geq k\big)\leq C\Big[\big( \epsilon^s L l^\beta+e^{-l/8}\big)^2e^{C\epsilon^s L l^{\rho'}}+e^{-L^\nu}\Big],
\end{displaymath}
where $l:=L^\alpha$.

\end{theo}
Theorem~\ref{minasmall}, Theorem~\ref{mina-D} and Theorem~\ref{mina3} are consequences of Theorem~\ref{mina} and Theorem~\ref{mina2}. For instance
Theorem~\ref{mina3} is a consequence of Theorem~\ref{mina2} if we take $\nu\in(\xi',\alpha)$. Indeed, it implies that $\epsilon\gtrsim L^{-k}\geq e^{-l^\nu}$ and that $\epsilon\geq e^{-l^\alpha}$.

The proof of  Theorem~\ref{mina} being the same as the proof of \cite[Theorem 1.1]{2011arXiv1101.0900K}, we only sketch it for the reader's convenience. First, the localization assumption implies that eigenvalues of $H_\omega(\Lambda)$ are close to eigenvalues of the same operator restricted to smaller boxes. Recall that  $\mathcal{U}_{\Lambda_L}$ is the set of $\omega$ such that (Loc)(I) holds.
\begin{lem}\label{loc}
For  $0<p<0$ and $\xi'<\xi<1$, there exists $L_{p,\xi,\xi'}>0$ such that for $l=(\log L)^{\frac{1}{\xi'}}$, $L\geq L_{p,\xi,\xi'}$, $\omega\in\mathcal{U}_{\Lambda_L}$ and $\gamma\in \Lambda_L$, if $H_\omega(\Lambda_L)$ has k eigenvalues in $[E-\epsilon,E+\epsilon]$ with localization center in $\Lambda_{\frac{4l}{3}}(\gamma)$, then $H_\omega( \Lambda_{\frac{3l}{2}}(\gamma))$ has k eigenvalues in $[E-\epsilon-e^{-l^{\xi}/8},E+\epsilon+e^{-l^{\xi}/8}]$.

\end{lem}

Define 
\begin{displaymath}
\Omega^b:=\left\{ \omega\in \mathcal{U}_{\Lambda_L}; 
\begin{aligned}
\text{there are two centres of localization of eigenfunctions}\\  
\text{associated to eigenvalues in }[E-\epsilon,E+\epsilon]\\
\text{that are at a distance less than 4l from each other}
\end{aligned}
\right\}
\end{displaymath}

and $\Omega^g:=\mathcal{U}_{\Lambda_L} / \Omega^b$.
 When $\omega\in\Omega^g$, Lemma~\ref{loc} gives rise to independent operators on small boxes : $(H_\omega( \Lambda_{\frac{3l}{2}}(\gamma)))_\gamma$. Hence, we obtain the following lemma :
\begin{lem}
Fix $0<p$ and $0<\xi'<\xi<1$. Then there exists $L_{p,\xi,\xi'}>0$ such that, for $l=(\log L)^{\frac{1}{\xi'}}$, for $L\geq L_{p,\xi,\xi'}$ and $k\geq 2$, one has
\begin{displaymath}
\mathbb{P}\big(\big \{ \omega \in \Omega^b; tr[\textbf{1}_{[E-\epsilon,E+\epsilon]}(H_\omega(\Lambda_L))]\geq k \big\} \big)\leq \dfrac{L}{l}\mathbb{P}_{2,9l,l}(\epsilon)+e^{-sl^{\xi}/9}
\end{displaymath}
and, for $k\leq [L/(4l)]+1$,
\begin{multline}
\mathbb{P}\big(\big \{ \omega \in \Omega^g; \text{tr}[\textbf{1}_{[E-\epsilon,E+\epsilon]}(H_\omega(\Lambda_L))]\geq k \big\} \big) \\
\leq  \begin{pmatrix}
[L/m]\\ k \end{pmatrix} 
(\mathbb{P}_{1,3l/2,4l/3}(\epsilon)+e^{-sl^{\xi}/8}\big)^k.
\end{multline}

where we have defined \begin{displaymath}
\mathbb{P}_{k,l,l'}:=\sup_{\gamma\in l'\Z\cap[O,L]} \mathbb{P}
\big( \text{tr}[\textbf{1}_{[E-\epsilon,E+\epsilon]}(H_\omega(\Lambda_l(\gamma)))]\geq k  \big).
\end{displaymath}
\end{lem}

This implies the following result : 

\begin{lem}
Fix J compact in $\mathcal{I}$, the localization region (cf (Loc) in Section 1). Then, for any $p>0$, $\xi\in(0,1)$ and $\xi'\in(0,\xi)$, there exists $C=C_{\xi,\xi'}>0$ and $L_{\xi,\xi'}>0$ such that, for $E\in J,L\geq L_{\xi,\xi'}$ and $\epsilon\in(0,1)$, one has
\begin{multline}
\sum_{k\geq 2} \mathbb{P}\big(\text{tr}[\textbf{1}_{[E-\epsilon,E+\epsilon]}(H_\omega(\Lambda_L))]\geq k\big)\leq  L^{-p}+\dfrac{L^2}{l}\mathbb{P}_{2,9l,l}(\epsilon)\\
+\left(\dfrac{L}{l}\right)^2\Big(\mathbb{P}_{1,3l/2,4l/3}(\epsilon)+e^{-l^{\xi}/8}\Big )^2e^{L\mathbb{P}_{1,3l/2,4l/3}(\epsilon)/l}
\end{multline}

where $l=(\log L)^{1/\xi'}$.

\end{lem}

Now, Theorem~\ref{mina} follows from Theorem~\ref{inv} and (W). As for Theorem~\ref{mina2}, it comes from the same reasoning and the following lemma the proof of which is similar to the proof of Lemma 3.4 in \cite{2011arXiv1101.0900K} except that we use (Loc)(II) instead of (Loc)(I)  :  

\begin{lem}\label{loc2}
For  $\alpha\in(0,1)$, $\xi\in(0,1)$ and $\nu\in(0,\alpha\xi)$, there exists $L_{\nu,\xi,\alpha}>0$ such that for $l=L^\alpha$ and $L\geq L_{\alpha,\xi,\xi'}$ and $\omega\in\mathcal{V}_{\Lambda_L}$ and $\gamma\in \Lambda_L$, if $H_\omega(\Lambda_L)$ has k eigenvalues in $[E-\epsilon,E+\epsilon]$ with localization center in $\Lambda_{\frac{4l}{3}}(\gamma)$, then $H_\omega( \Lambda_{\frac{3l}{2}}(\gamma))$ has k eigenvalues in $[E-\epsilon-e^{-l^{\xi}/8},E+\epsilon+e^{-l^{\xi}/8}]$.

\end{lem}

\section{Decorrelation estimates for eigenvalues }\label{secdec1}\indent 
Now that we have proven Minami estimates, we can now turn to prove decorrelation estimates for eigenvalues. The Minami estimates and Theorem~\ref{deco} below are complementary results. Indeed, Minami estimates can be seen as decorrelation estimates for close eigenvalues whereas Theorem~\ref{deco} is proven for distinct energies. In fact, as the proof will show, Theorem~\ref{deco} is a consequence of Theorem~\ref{mina-D} and localization. To prove decorrelation estimates, we will study the co-linearity of gradients of eigenvalues, as functions of the random variables. As the gradients are different according to the models considered, the co-linearity condition depends on the model. Hence, we divide the proof of decorrelation estimate into three parts, according to the model.

\subsection{General decorrelation estimates} \label{subsecdec1}

In this section, we prove decorrelation estimate when $V_\omega(n)=\omega_n$ so that $(V_\omega(n))_n$ are independent variables. Independence will only be used to simplify the estimations of probabilities and we show in Subsection~\ref{alloy} how to prove decorrelation estimates for the discrete alloy-type model. Most results of this section only require (IAD). The result is slightly different from the decorrelation estimate found in \cite[Lemma 1.1]{K10}. Indeed we state a result that will handle both discrete alloy-type model and unidimensional quantum graphs with random vertex coupling. Therefore we introduce the following notations.

For $E\in\R$, pick  $\lambda_E\in\R$, $\mu_E\in\R$ and define $\sigma:=\{E\in\R,\lambda_E\neq0\}$. For  $E\neq E'$, define $d:=\lambda_E+\lambda_{E'}$, $b:= \lambda_E-\lambda_{E'}$, $c:=\mu_E-\mu_{E'}$ and for $E\in\sigma$, $V_\omega(E):=\dfrac{1}{\lambda_E}(V_\omega-\mu_E)$. In (W) we assume that $s=1$. The results of the present section need the following hypothesis to hold: $(E,E')\in\sigma^2$ must be taken such that at least two of d,b,c be non-zero. If $\lambda_E=1$ and $\mu_E=E$ we have $\sigma=\R$ and $V_\omega(E)=V_\omega-E$. In this case, we always have $d=2,b=0,c\neq 0$.

In this section, the sequence $a(n):=(a_\omega(n))$ is not random. Thus, the operator $\Delta_a$ is deterministic.
We follow the lines of \cite{K10} to prove a decorrelation estimate for the eigenvalues. 
If the random variables $(V_\omega(n))_n$ are independent we can use the following Minami estimates (see for instance \cite{CGA09}) which are more precise than the estimates in Theorem~\ref{mina-D}.

\begin{theo}[\textbf{M}]
Let $I\subset J\subset \R$, then for any $E\in\sigma$ there exists $C>0$ such that
\begin{equation}
\mathbb{E}\big(tr\textbf{1}_I(-\Delta_a+V_\omega(E))_{\Lambda_l}(tr\textbf{1}_J(-\Delta_a+V_\omega(E))_{\Lambda_L}-1)\big) \leq C |I||J||\Lambda|^2.
\end{equation}
\end{theo}

 The main result of the present section is : 

\begin{theo}[\textbf{D}]\label{deco}
Let $\beta\in(1/2,1)$. For $\alpha\in (0,1) $ and $(E,E')\in\sigma^2$ such that at least two of $d,b,c$ be non zero, then for any $k>0$ there exists $C>0$ such that for $L$ sufficiently large and $kL^\alpha\leq l\leq L^\alpha/k$ we have.
\begin{displaymath}
\mathbb{P}\left(
\begin{aligned}
 tr\textbf{1}_{[-L^{-1},+L^{-1}]}(-\Delta_a+V_\omega(E)\restriction{)}{\Lambda_l}\neq 0,\\
 tr\textbf{1}_{[-L^{-1},+L^{-1}]}(-\Delta_a+V_\omega(E')\restriction{)}{\Lambda_l}\neq 0
\end{aligned}
\right)\leq C\left(\dfrac{l^2}{L^{4/3}}\right) e^{(\log L)^\beta}.
\end{displaymath}
\end{theo}

In particular, this proves (D) (see Section 2) for the operator $H_\omega:=-\Delta_a+ V_\omega(n)$ with $V_\omega(n)=\omega_n$. 

Using (M) or Theorem~\ref{mina-D} if we just have (IAD), as $(\log L)^C <e^{(\log L)^\beta}$, Theorem~\ref{deco} is a consequence of the following theorem : 

\begin{theo}\label{thdec}
Let $\beta\in(1/2,1)$. For $\alpha\in (0,1) $ and $(E,E')\in\sigma^2$ such that at least two of $d,b,c$ be non zero, then for any $k>1$ there exists $C>0$, such that for $L$ large enough and $kL^\alpha\leq l\leq L^\alpha/k$ we have
\begin{displaymath}
\mathbb{P}_0:=\mathbb{P}\left(
\begin{aligned}
  tr\textbf{1}_{[-L^{-1},+L^{-1}]}(-\Delta_a+V_\omega(E)\restriction{)}{\Lambda_l}= 1,\\
 tr\textbf{1}_{[-L^{-1},+L^{-1}]}(-\Delta_a+V_\omega(E')\restriction{)}{\Lambda_l}= 1
\end{aligned}
\right)\leq C\left(\dfrac{l^2}{L^{4/3}}\right)e^{(\log L)^\beta}.
\end{displaymath}
\end{theo}

Now, using (Loc)(I) , Theorem~\ref{thdec} is a consequence of the following theorem : 
\begin{theo}\label{thdec2}
Let $\beta'\in(1/2,1)$.For $\alpha\in (0,1) $ and $(E,E')\in\sigma^2$ such that at least two of $d,b,c$ be non zero, then there exists $C>0$ such that for any $\xi'\in(0,\xi)$, $L$ large enough and $\tilde{l}=(\log L)^{1/\xi'}$ we have 
\begin{displaymath}
\mathbb{P}_1:=\mathbb{P} \left( 
\begin{aligned}
 tr\textbf{1}_{[-2L^{-1},+2L^{-1}]}(-\Delta_a+V_\omega(E)\restriction{)}{\Lambda_{\tilde{l}}}= 1,\\
 tr\textbf{1}_{[-2L^{-1},+2L^{-1}]}(-\Delta_a+V_\omega(E')\restriction{)}{\Lambda_{\tilde{l}}}= 1
\end{aligned}
\right)\leq C\left (\dfrac{\tilde{l}^2}{L^{4/3}}\right)e^{\tilde{l}^{\beta'}}.
\end{displaymath}
\end{theo}
\begin{proof}[Proof of Theorem~\ref{thdec} using Theorem~\ref{thdec2}]
\text{ } \\
Define $H^E_\omega(\Lambda_l):=-(\Delta_a+V_\omega(E))_{\Lambda_l}$,  $H^{E'}_\omega(\Lambda_l):=-(\Delta_a+V_\omega(E'))_{\Lambda_l}$ and $\tilde{J_L}=[-2L^{-1},2L^{-1}]$. Suppose that both operators have one eigenvalue in $[-L^{-1},L^{-1}]$. For $\omega$ in $\mathcal{U}_{\Lambda_L}$(cf (Loc)(I)), to each eigenvalue we can associate  a localization centre. Using Lemma~\ref{loc} with $p=2$, depending on whether these points are distant from at least $3\tilde{l}$ or are at a distance of at most $3\tilde{l}$ , either there exists $\gamma\neq\gamma'$ such that $\textbf{dist}(\gamma+\Lambda_{\tilde{l}};\gamma'+\Lambda_{\tilde{l}})\geq \tilde{l}$  with 
\begin{enumerate}
\item $H_\omega(\gamma+\Lambda_{\tilde{l}})$ has exactly one eigenvalue in $\tilde{J_L}$
\item $H_\omega(\gamma'+\Lambda_{\tilde{l}})$ has exactly one eigenvalue in $\tilde{J'_L}$
\end{enumerate}
or there exists $\gamma_0$ such that $H_\omega(\gamma_0+\Lambda_{5\tilde{l}})$ has exactly one eigenvalue in $\tilde{J_l}$ and exactly one in $\tilde{J'_L}$. Using (IAD) and (W), we then compute
\begin{align*}
\mathbb{P}_0&\leq L^{-2}+C(l/\tilde{l})\mathbb{P}
\left(\left\{
\begin{aligned}
\sigma(H_\omega(\Lambda_{5\tilde{l}}(\gamma_0)\cap \tilde{J_L}\neq \emptyset \\ 
\sigma(H_\omega(\Lambda_{5\tilde{l}}(\gamma_0)\cap \tilde{J'_L}\neq \emptyset 
\end{aligned}
\right\}\right)\\
&\;\;\;\;\;\;+C(l/\tilde{l})^2\mathbb{P}(\sigma(H_\omega(\Lambda_{\tilde{l}}(\gamma)\cap \tilde{J_L}\neq \emptyset)\mathbb{P}(\sigma(H_\omega(\Lambda_{\tilde{l}}(\gamma')\cap \tilde{J'_L}\neq \emptyset)\\
&\leq L^{-2}+C(l/\tilde{l})^2(\tilde{l}^\rho/L)^2+C(l/\tilde{l})\mathbb{P}_1\leq c(l/L)^2\tilde{l}^{2(\rho-1)}+C(l/\tilde{l})\mathbb{P}_1.
\end{align*}
Now, if one take $\xi$ and $\xi'$ sufficiently close to 1, we obtain Theorem~\ref{thdec} with $\beta=\dfrac{\beta'}{\xi}\in\left(\dfrac{1}{2},1\right)$.
\end{proof}

We can now turn to the proof of Theorem~\ref{thdec2}. From now on we will write $H_\omega(\Lambda_l,E)=(-\Delta_a+V_\omega(E)\restriction{)}{\Lambda_{\tilde{l}}}$ and $J_L=[-L^{-1},L^{-1}]$. For $\epsilon\in(2L^{-1},1)$, for some $\kappa>2$, using Theorem~\ref{mina2} when one of the two operators $H_\omega(\Lambda_l,E)$ and $H_\omega(\Lambda_l,E')$ has two eigenvalues in $[-\epsilon,+\epsilon]$, one has 
\begin{equation}
\mathbb{P}_1\leq C\epsilon^2 l^{\kappa}+\mathbb{P_\epsilon}\leq C\epsilon^2l^2e^{l^\beta}+\mathbb{P}_\epsilon
\end{equation}
where 
\begin{displaymath}
\mathbb{P_\epsilon}=\mathbb{P}(\Omega_0(\epsilon))
\end{displaymath}
and

\[ \Omega_0(\epsilon)= \left\{ \omega;
\begin{aligned}
\sigma(H_\omega(\Lambda_l,E))\cap J_L&= \{E(\omega)\} \\
\sigma(H_\omega(\Lambda_l,E))\cap (-\epsilon,&+\epsilon)= \{E(\omega)\} \\
\sigma(H_\omega(\Lambda_l,E'))\cap J_L&= \{E'(\omega)\} \\
\sigma(H_\omega(\Lambda_l,E'))\cap(-\epsilon,&+\epsilon)= \{E'(\omega)\}
\end{aligned} 
\right \}.
\]
In order to estimate $\mathbb{P}_\epsilon$ we make the following definition. For $(\gamma,\gamma')\in\Lambda_L^2$ let $J_{\gamma,\gamma'}(E(\omega),E'(\omega))$ be the Jacobian of the mapping $(\omega_\gamma,\omega_{\gamma'})\rightarrow (E(\omega),E'(\omega))$ : 
\begin{equation}\label{defjac}
J_{\gamma,\gamma'}(E(\omega),E'(\omega))=\left \vert \begin{pmatrix} \partial_{\omega_\gamma}E(\omega) & \partial_{\omega_{\gamma'}}E(\omega)\\ \partial_{\omega_\gamma}E'(\omega) &\partial_{\omega_{\gamma'}}E'(\omega)
\end{pmatrix} \right \vert
\end{equation} 
and define 
\begin{displaymath}
\Omega^{\gamma,\gamma'}_{0,\beta}(\epsilon)= \Omega_0(\epsilon)\cap \left \{ \omega ;|J_{\gamma,\gamma'}(E(\omega),E'(\omega))|\geq \lambda \right\}.
\end{displaymath} 

When one of the Jacobians is sufficiently large, the eigenvalues depends on  two independent random variables. Thus the probability to stay in a small interval is small. So we divide the proof in two parts, depending on whether all the Jacobians are small. The next lemma shows that if all the Jacobians are small then the gradients of the eigenvalues must be almost co-linear.

\begin{lem}\label{grad->jac}
Let $(u,v)\in(\R^+)^{2n}$ such that $\|u\|_1=\|v\|_1=1$. Then 
\begin{displaymath}
\max_{j\neq k} \left | \begin{pmatrix} u_j & u_k\\ v_j & v_k \end{pmatrix} \right |^2\geq \dfrac{1}{4n^5}\Vert u-v \Vert_1^2.
\end{displaymath}
\end{lem}

So either one of the Jacobian determinants is not small or the gradient of $E$ and $E'$ are almost collinear. We shall show that the second case happens with a small probability. But first we show that a normalized eigenfunction has a large number of sites where it is not small.

\begin{lem}\label{lem-large}
There exists $c>0$ such that for $u$ a normalized eigenvalue of \\ $(-\Delta_a+V_\omega(E)\restriction{)}{\Lambda_l}$, there exists $n_0\in \Lambda_l$ such that for $m\in I:=[n_0-cl^\beta,n_0+cl^\beta]$ either $|u(m)|\geq e^{-l^\beta}$ or $|u(m+1)|\geq e^{-l^\beta}$.
\end{lem}
\begin{proof}
\normalfont
As $u$ satisfies a finite difference equation of order two, u grows at most exponentially fast, so if $u$ takes two small consecutive values, $u$ cannot be large in a logarithmic neighbourhood of these points. As u is normalized there exists $n_0\in\Lambda_l$ such that $|u(n_0)|\geq 1/\sqrt{l}$
. \\
If there exists $m\in\{1,\dots,l\}$ such that $|u(m)|\leq e^{- l^\beta/5}$ and $|u(m+1)|\leq e^{- l^\beta/5}$ then as 
\begin{align*}
|u(n+2)|&\leq C(|u(n+1)|+|u(n)|)\\
|u(n)|&\leq C(|u(n+1)|+|u(n+2)|)
\end{align*}
one has $|u(m+N)|\leq (2C)^{|N|}e^{-l^\beta/5}$. So $|m-n_0|\geq  \tilde{c}l^\beta$ for some $\tilde{c}>0$ and for some $0<c<\tilde{c}$ and m such that $|m-n_0|\leq c l^\beta$ either 
$|u(m)|\geq e^{-l^\beta/5}$ or $|u(m+1)|\geq e^{-l^\beta/5}$. So there are at least $cl^\beta/2$ sites where u is not small. This proves Lemma~\ref{lem-large}.
\end{proof}

 We are now able to give a proof of the 
  
\begin{lem}\label{probcoli}
Let $E,E'\in\sigma^2$ and $\beta>1/2$. Let $\mathbb{P}$ denote the probability that there exist $E_j(\omega)$ and $E_k(\omega)$, simple eigenvalues of $(-\Delta+V_\omega(E))_{\Lambda_l}$ and $(-\Delta+V_\omega(E'))_{\Lambda_l}$ in $[-e^{-l^\beta},+e^{-l^\beta}]$ such that
\begin{equation}
\left\|\dfrac{\nabla_\omega\big(E_j(\omega))}{\|\nabla_\omega\big(E_j(\omega))\|_1}\pm\dfrac{\nabla_\omega\big(E_k(\omega))}{\|\nabla_\omega\big(E_k(\omega))\|_1}\right\|_1\leq e^{-l^\beta}.
\end{equation}
Then there exists $c>0$ such that
\begin{equation}
\mathbb{P}\leq e^{-c l^{2\beta}}
\end{equation}
\end{lem}
\begin{proof}
\normalfont
The proof follows that of Lemma 2.4 in \cite{K10}.
Let $\phi_j$ be a normalized eigenfunction associated to $E_j$. Then we compute
\begin{align*}
\partial_{\omega_n} E_j(\omega)&=\left\langle \left(\partial_{\omega_n}H_\omega\right) \phi_j,\phi_j\right\rangle\\
&=\left\langle \lambda_E \phi_j(n)e_n,\phi_j\right\rangle=\lambda_E \phi_j^2(n)
\end{align*}
Thus,
$\dfrac{\nabla_\omega(E_j(\omega))}{\|\nabla_\omega(E_j(\omega))\|_1}(n)=\dfrac{\lambda_E}{|\lambda_E|}\phi_j(n)^2=\pm \phi_j(n)^2$. Hence, as in \cite{K10}, there exist $\mathcal{P},\mathcal{Q}$ with $\mathcal{P}\cap\mathcal{Q}=\emptyset$ and $\mathcal{P}\cup\mathcal{Q}=\{1,\dots,L\}$ such that
\begin{align*}
	\forall n\in\mathcal{P}, |\phi_j(n)-\phi_k(n)|\leq e^{-l^\beta/2}\text{ ,}\\
	\forall n\in\mathcal{Q}, |\phi_j(n)+\phi_k(n)|\leq e^{-l^\beta/2}\text{ .}
\end{align*}
Let P (respectively Q) be the orthogonal projector on $\ell^2(\mathcal{P})$ (respectively $\ell^2(\mathcal{Q})$) in $\ell^2(\{1,\dots,l\})$. Then, the eigenvalues equations can be rewritten as 
\begin{align*}
	(\lambda_E\Delta_a+\mu_E+V_\omega)(Pu+Qu)=h_E\text{ ,}\\
	(\lambda_{E'}\Delta_a+\mu_{E'}+V_\omega)(Pu-Qu)=h_{E'}\text{ .}
\end{align*}
where $|h_E|+|h_{E'}|\leq Ce^{-l^\beta}$.
Multiplying by $P$ and $Q$ and summing or subtracting we obtain : 
\begin{multline} 
(\lambda_E+\lambda_{E'})\Delta_a Pu+(\lambda_E-\lambda_{E'})\Delta_a Qu\\
+(\mu_E+\mu_E')Pu+(\mu_E-\mu_E')Qu+2V_\omega Pu=h_E+h_{E'}\text{ ,}
\end{multline}
\begin{multline}
(\lambda_E-\lambda_{E'})\Delta_a Pu+(\lambda_E+\lambda_{E'})\Delta_a Qu+\\
(\mu_E-\mu_E')Pu+(\mu_E+\mu_E')Qu+2V_\omega Qu=h_E-h_{E'}\text{ .}
\end{multline}
		
As $PV_\omega Q=QV_\omega P=0$, we have
\begin{multline}\label{eigeq1}
(\lambda_E+\lambda_{E'})(Q\Delta_a P+P \Delta_a Q)u\\+(\lambda_E-\lambda_{E'})(Q\Delta_a Q+P\Delta_a P) u+(\mu_E-\mu_E')u=h_{E,E'}\text{ .}
\end{multline}
Recall that if one of $b:=(\lambda_E-\lambda_{E'})$, $c:=\mu_E-\mu_E'$, $d:=(\lambda_E+\lambda_{E'})$ is zero then the other two are not. In \cite{K10} is proven that $P \Delta Q + Q\Delta P=\sum_k C_k\Delta C_k$ where the sets $(C_k)_k$ are "intervals" in $\mathbb{N}/ l \mathbb{Z}$. And as we have
\begin{equation*}
P\Delta_a P+Q\Delta_a Q+ P \Delta_a Q + Q\Delta_a P = \Delta_a,
\end{equation*}
if $\mathcal{C}=\cup C_k$, $P\Delta_a P+Q\Delta_a Q=\sum_k D_k\Delta_a D_k$ where 
$\mathcal{D}=\cup_k D_k =  \{1,\dots, l\}-\cup_k \mathring{C_k} $. So $\mathcal{C}$ and $\mathcal{D}$ have symmetric roles in the equation. \eqref{eigeq1} is equivalent to the following equation : 
\begin{equation}\label{combeq}
(b\Delta_a+2\lambda_E'\sum_k C_k\Delta_a C_k+c)u=h.
\end{equation}

If $n\notin \mathcal{C}$ then 
\begin{equation}\label{ext}
b\Big(a(n+1)u(n+1)+a(n)u(n-1)\Big)+c u(n) = h(n),
\end{equation}
so if $b=0$ then $c\neq 0$ and $|u(n)|\lesssim e^{-l^\beta}$. Suppose that $b=0$. Then, we know from Lemma~\ref{lem-large} that $\mathcal{C}^c$ cannot have two consecutive points in $I$ (defined in Lemma~\ref{lem-large})  and same for $\mathcal{D}^c$ if $d=0$. So, without loss of generality, we can suppose that $d\neq 0$ and that $\sharp (\mathcal{C}\cap I)\geq \sharp I/2$ .

If $n\in \mathring{C_l}$ then
\begin{equation}\label{int}
d\Big(a(n+1)u(n+1)+a(n)u(n-1)\Big)+c u(n) = h(n).
\end{equation}

If $C_l=\{n_-,\dots,n_+\}$ then

\begin{align}\label{bord}
b a(n_++1)u(n_++1)+da(n_+) u(n_+-1)+cu(n)=h(n),\\
d a(n_-+1)u(n_-+1)+b a(n_-) u(n_--1)+cu(n)=h(n).
\end{align}
Now, we show that this implies stringent conditions on the random variables $(\omega_n)_n$. Indeed, we know that the eigenvalue equation for u is 
\begin{equation}\label{eig}
\lambda_E\Big(a(n+1)u(n+1)+a(n)u(n-1)\Big)+(\mu_E + \omega_n)u(n)=h_E(n).
\end{equation}
Combining the eigenvalue equation with (\ref{int}) we obtain
\begin{align}
\forall n\in \mathring{\mathcal{C}_k},\, \left|\mu_E + \omega_n-\lambda_E \dfrac{c}{d} \right||u(n)|\leq e^{-L^\beta}.
\end{align}
So if $|u(n)|\geq e^{-l^\beta/2}$ then we have

\begin{align}\label{cond1}
 \left|\mu_E + \omega_n+\lambda_E \dfrac{c}{d}\right|\leq e^{-l^\beta/2}.
\end{align}

Suppose $\mathcal{C}_k\subset I$. If $|\mathcal{C}_k|=n_k$ there are at least $\left\lfloor \dfrac{n_k}{2}-1\right\rfloor$ sites in $\mathring{\mathcal{C}}_k$ where u is not small. We are now left with the study of the cases $n_k=2$ and $n_k=3$. Suppose $b\neq0$. Combining (\ref{bord}) and (\ref{eig}), we have
\begin{equation}
\left|\lambda_E\left(1-\dfrac{d}{b}\right)a(n_+)u(n_+-1)+\left(\mu_E + \omega_{n_+}-\lambda_E \dfrac{c}{b}\right)u(n_+)\right|\leq Ce^{-l^\beta},
\end{equation}
\begin{equation}
\left|\lambda_E\left(1-\dfrac{d}{b}\right)a(n_-+1)u(n_-+1)+\left(\mu_E + \omega_{n_-}-\lambda_E \dfrac{c}{b}\right)u(n_-)\right|\leq Ce^{-l^\beta}.
\end{equation}

When $n_k=2$ we have 

$\left \Vert \begin{pmatrix}
\mu_E + \omega_{n_+}-\lambda_E \dfrac{c}{b}&\lambda_E\left(1-\dfrac{d}{b}\right)a(n_+)\\
\lambda_E\left(1-\dfrac{d}{b}\right)a(n_+)&\mu_E + \omega_{n_-}-\lambda_E \dfrac{c}{b}
\end{pmatrix}
\begin{pmatrix}
u(n_+)\\
u(n_-)
\end{pmatrix} \right \Vert \leq Ce^{-l^\beta}$.
\\
As $\left \Vert \begin{pmatrix} u(n_-)\\u(n_+)\end{pmatrix} \right \Vert \geq Ce^{-l^\beta/2}$ we have 
\begin{equation}\label{cond2}
\left|\left(\mu_E + \omega_{n_+}-\lambda_E \dfrac{c}{b}\right)\left(\mu_E + \omega_{n_-}-\lambda_E \dfrac{c}{b}\right)-\lambda_E^2\left(1-\dfrac{d}{b}\right)^2a(n_+)^2\right|\leq e^{-l^\beta/2}.
\end{equation}

If $n_k=3$ we have $\mathcal{C}_k=\{n_-,n,n_+\}$ and
\begin{multline*}
\left \Vert \begin{pmatrix}
\mu_E + \omega_{n_+}-\lambda_E \dfrac{c}{b}&\lambda_E\left(1-\dfrac{a}{b}\right)a(n_+)&0\\
0&\mu_E + \omega_n-\lambda_E \dfrac{c}{a}&0\\
0&\lambda_E\left(1-\dfrac{a}{b}\right)a(n)&\mu_E + \omega_{n_-}-\lambda_E \dfrac{c}{b}
\end{pmatrix} 
\begin{pmatrix}
u(n_+)\\
u(n)\\
u(n_-)
\end{pmatrix}
\right \Vert \\
\leq Ce^{-l^\beta}.
\end{multline*}
If $|u(n)|\geq e^{-l^\beta/2}$ then  
\begin{equation}\label{cond3}
\left|\mu_E + \omega_n-\lambda_E \dfrac{c}{a}\right|\leq e^{-l^\beta/2}.
\end{equation}
If $|u(n)|\leq e^{-l^\beta/2}$ then $|u(n_+)|\geq e^{-l^\beta/2}$ and  
 \begin{equation}
  |\mu_E + \omega_{n_+}-\lambda_E \dfrac{c}{b}|\leq Ce^{-l^\beta/2}.
\end{equation}

Now, suppose $b=0$. If there is a gap between $C_k$ and its two neighbours $\{C_{k+1},C_{k-1}\}$, we know that $|u(n_+ +1)|\leq e^{-l^\beta}$ and $|u(n_- -1)|\leq Ce^{-l^\beta}$.
So the conditions when $n_k\in\{2,3\}$ are 
\\
$\left \Vert \begin{pmatrix}
\mu_E + \omega_{n_+}&\lambda_Ea(n_+)\\
\lambda_Ea(n_+)&\mu_E + \omega_{n_-}
\end{pmatrix}
\begin{pmatrix}
u(n_+)\\
u(n_-)
\end{pmatrix} \right \Vert \leq e^{-l^\beta}
$ if $n_k=2$ and
 \\
 \text{ }\\
 $\left \Vert \begin{pmatrix}
\mu_E + \omega_{n_+}&\lambda_Ea(n_+)&0\\
0&\mu_E + \omega_n-\lambda_E\dfrac{c}{a}&0\\
0&\lambda_Ea(n)&\mu_E + \omega_{n_-}
\end{pmatrix} 
\begin{pmatrix}
u(n_+)\\
u(n)\\
u(n_-)
\end{pmatrix}
\right \Vert\leq Ce^{-l^\beta}$
if $n_k=3$.\\
In both cases we have a condition on the $(\omega_n)_n$ similar to \eqref{cond1} and \eqref{cond2}.

If $C_k$ is adjacent to $C_{k+1}$ it gives rise to an interval $D_{k_0}$ in $\mathcal{D}$ of length two (which makes the junction between $C_k$ and $C_{k+1}$) and we still have a condition on the $(\omega_n)_n$ similar to \eqref{cond2}, where $d$ and $b$ are exchanged.
To summarize, if $C_l\subset I$, when $n_k\geq4$ we have at least $\lfloor(n_k-2)/2\rfloor$ conditions of type (\ref{cond1}) and if $n_k=2$ or $n_k=3$ we have one condition of type (\ref{cond2}) or (\ref{cond3}), possibly shared with one of its neighbours.
 
There are at most two $\mathcal{C}_l$ that are not entirely contained in $I$, and for these particular two sets we only have $\big \lfloor(\sharp (\mathcal{C}_l\cap I)-1)/2\big \rfloor$ points where u is not small. So if $\sharp (\mathcal{C}_l\cap I)\in\{1,2\}$ there is no known point where u satisfies $|u(n)|\geq e^{-^\beta/2}$. Define $\tilde{n}_k=\sharp (\mathcal{C}_k\cap I)$, we know that $\sum_k\tilde{n}_k \geq cl^\beta/2$.
 Then, we have a number of conditions on the $(\omega_n)_n$ which is greater than  
\begin{align*}
\sum_{k,\tilde{n}_k\geq 3 ,\tilde{n_k}<n_k
} \lfloor (\tilde{n_k}-1)/2 \rfloor &+  \sum_{k, \tilde{n}_k\geq 4, \tilde{n}_k=n_k} \lfloor(n_k-2)/2\rfloor 
+ \frac{1}{2}\sum_{k,n_k\in\{2,3\},\tilde{n_k}=n_k} 1  \\ 
\geq& \sum_{k, \tilde{n}_k\geq 3, \tilde{n_k}<n_k} \tilde{n}_k/5+\sum_{k, \tilde{n}_k\geq 3 , \tilde{n_k}<n_k} \tilde{n}_k/3 \\
&+\sum_{k, \tilde{n}_k\geq 4 ,\tilde{n}_k=n_k} n_k/5+\frac{1}{2} \sum_{k, n_k\in\{2,3\},\tilde{n_k}=n_k}n_k/5-2 \\ 
\geq& ~cl^\beta/10-2\geq c'l^\beta/10
\end{align*}
 
There are $2^l$ choices for $\mathcal{P}$, $l$ choices for $n_0$, less than $2^l$ choices for the points where $u$ satisfies $|u(n)|\geq e^{-^\beta/2}$. For each of these choices, it gives us a probability at most $e^{-cl^{2\beta}}$. Thus, as $\beta>1/2$, we have for a smaller c and L sufficiently large
 \begin{equation}\label{proba1}
 \mathbb{P}\leq e^{-cl^{2\beta}}\leq L^{-2}.
\end{equation}

This complete the proof of Lemma~\ref{probcoli}.

\end{proof}

Pick $\lambda=e^{-l^\beta}\lambda_E\lambda_{E'}$. Then, either one of the Jacobian determinant is larger than $\lambda$ or the gradients are almost co-linear. Lemma~\ref{probcoli} shows that the second case happens with a probability at most $e^{-cL^{2\beta}}$. It remains to evaluate $\mathbb{P}(\Omega^{\gamma,\gamma'}_{0,\beta}(\epsilon))$. We recall the following results from \cite{K10}. They were proved for $a_n=1$, they extend readily to our case. First, we study the variations of the Jacobian. 

\begin{lem}\label{Hessien}

There exists $C>0$ such that
\begin{displaymath}
\Vert Hess_\omega(E(\omega))\Vert_{l^\infty\rightarrow l^1}\leq \dfrac{C}{dist\big[E(\omega),\sigma(H_\omega(\Lambda_l))-\{E(\omega)\}\big ]}.
\end{displaymath}

\end{lem}

Fix $\alpha\in(1/2,1)$. Using Lemma~\ref{Hessien} and Theorem~\ref{mina-D} when $H_\omega(\Lambda_l)$ has two eigenvalue in $[E-L^{-\alpha},E+L^{-\alpha}]$, for L large enough, with probability at least $1-L^{-2\alpha}\lambda $,
\begin{equation}
\Vert Hess_\omega(E(\omega))\Vert_{l^\infty\rightarrow l^1}+\Vert Hess_\omega(E'(\omega))\Vert_{l^\infty\rightarrow l^1}\leq CL^\alpha
\end{equation}
 
In the following lemma we write $\omega=(\omega_\gamma,\omega_{\gamma'},\omega_{\gamma,\gamma'})$.
\begin{lem}\label{square}
Pick $\epsilon= L^{-\alpha}$. For any $\omega_{\gamma,\gamma'}$, if there exists $(\omega_\gamma^0,\omega_{\gamma'}^0)\in \R^2$ such that $(\omega_\gamma^0,\omega_{\gamma'}^0,\omega_{\gamma,\gamma'})\in\Omega^{\gamma,\gamma'}_{0,\beta}(\epsilon)$, then for $(\omega_\gamma,\omega_{\gamma'}) \in \R^2$ such that $|(\omega_\gamma,\omega_{\gamma'})-(\omega_\gamma^0,\omega_{\gamma'}^0)|_\infty\leq \epsilon$ one has 
\begin{displaymath}
(E_j(\omega),E_k(\omega))\in J_L\times J_L'\Longrightarrow |(\omega_\gamma,\omega_{\gamma'})-(\omega_\gamma^0,\omega_{\gamma'}^0)|_\infty\leq L^{-1}\lambda^{-2}.
\end{displaymath} 
\end{lem}

As in Lemma~\ref{square}, fix $(\omega_\gamma^0,\omega_{\gamma'}^0)$ such that $(\omega_\gamma^0,\omega_{\gamma'}^0,\omega_{\gamma,\gamma'})\in\Omega^{\gamma,\gamma'}_{0,\beta}(\epsilon)$ and define $\mathcal{A}:=(\omega_\gamma^0,\omega_{\gamma'}^0)+\{(\omega_\gamma,\omega_{\gamma'}) \in \R_+^2\cup \R_-^2 ,\left|\omega_\gamma\right|\geq \epsilon \text{ or } \left|\omega_{\gamma'}\right|\geq \epsilon \}$. We know that for any $i\in\Z$, $\omega_i\rightarrow E_j(\omega)$ and  $\omega_i\rightarrow E_k(\omega)$ are non increasing functions. Thus, if $ 
(\omega_\gamma,\omega_{\gamma'})\in\mathcal{A}$ then $(E_j(\omega),E_k(\omega))\notin J_L\times J_L'$. Thus, all the squares of side $\epsilon$ in which there is a point in $\Omega^{\gamma,\gamma'}_{0,\beta}(\epsilon)$ are placed along a non-increasing broken line that goes from the upper left corner to the bottom right corner. As the random variables are bounded by $C>0$, there is at most $C L^\alpha$ cubes of this type.

As the $(\omega_n)_n$ are i.i.d, using Lemma~\ref{square} in all these cubes, we obtain  : 

\begin{equation}\label{proba2}
\mathbb{P}(\Omega^{\gamma,\gamma'}_{0,\beta}(\epsilon))\leq CL^{\alpha-2}\lambda^{-4}
\end{equation}

and therefore

\begin{equation}
\mathbb{P}_\epsilon\leq CL^{\alpha-2}\lambda^{-3}
\end{equation}

Optimization yields $\alpha=2/3$. This completes the proof of Theorem~\ref{thdec2}.

\subsection{Decorrelation estimates for alloy-type models }\label{alloy}

In Subsection~\ref{subsecdec1} we proved decorrelation estimates under the assumptions that the random potential was taking independent values. In this section we prove decorrelation estimates for alloy-type potentials. Thus, the values taken by the random potential are correlated. 
Let $L\in\N$ and $(\omega_n)_n$ be independent random variables with a common compactly supported bounded density $\rho$.
Let $(d_n)_{n\in\Z}\in\R^{\Z}$ and define $H_\omega(\Lambda_L)=\Delta_a+V_\omega$ on $\ell^2(\{1,\dots,L\})$ with periodic conditions where \begin{displaymath}
V_\omega(m)=\sum_{n\in\Z} a_n \omega_{n+m}
\end{displaymath}
From now on we will suppose that there exists $K>0$ and $\lambda>0$ such that $|d(n)|\leq K e^{-\lambda|n|}$. Without loss of generality, we can suppose that $d(0)$ is the right most maximum of the sequence, i.e :  
\begin{align}\label{cond_a1}
&|d(0)|=\textbf{max}\{|d(n)|,n\in\Z\}\notag,\\
&\forall n\geq 1, |d(n)|<|d(0)|.
\end{align}

Let $S=C\log L$ and define 
\begin{equation}\label{defzn}
z_m:=\tilde{V_\omega}(m):=\sum_{|n|\leq S}d_n \omega_{n+m}
\end{equation}

 and $\tilde{H_\omega}:=\Delta+\tilde{V_\omega}$.
Let $\kappa>1$, then there exists $C=C_{\lambda,K,\kappa}$ such that
\begin{equation}
\|\tilde{H_\omega}-H_\omega\|=\|\tilde{V_\omega}-V_\omega\|_\infty\lesssim e^{-\lambda S}\lesssim e^{-\lambda S}\lesssim L^{-\kappa}
\end{equation}
for C sufficiently large. Hence, one has the following relation : 
\begin{equation}\label{cut}
\begin{split}
\textbf{tr}\big(1_{[E-L^{-\kappa},E+L^{-\kappa}]}({H_\omega}(\Lambda))\big)\leq 
\textbf{tr}\big(1_{[E-2L^{-\kappa},E+2L^{-\kappa}]}(\tilde{H_\omega}(\Lambda))\big) \\
\leq \textbf{tr}\big(1_{[E-3L^{-\kappa},E+3L^{-\kappa}]}({H_\omega}(\Lambda))\big)
\end{split}
\end{equation}

Now we know that if $(\Lambda_j)_j$ is a collection of cubes distant from one another of at least $2S$, then the operators $(\tilde{H_\omega}(\Lambda_j))_j$ are independent. In \cite{springerlink:10.1007/s00023-010-0052-5} a Wegner estimate is proven for the operator $H_\omega$ under the hypothesis that the density of $(\omega_n)_n$ have bounded variation.

\begin{theo}[\cite{springerlink:10.1007/s00023-010-0052-5}]\label{wegalloy}

There exists $D \in\N, D \neq 0$, depending only on $(a_n)_n$ such that for each $\beta> D/\lambda$ there exists $K>0$ such that for all $L\in\N, E\in\R$ and $\epsilon>0$
\begin{equation}
\mathbb{E}\left[\textbf{tr}\big(1_{[E-\epsilon,E+\epsilon]}(H_\omega(\Lambda_L))\big)\right] \leq C \epsilon L(L+\beta \log L+K)^{D+1}.
\end{equation}
\end{theo}

Theorem~\ref{wegalloy} yields the following Wegner estimates, for L sufficiently large   
\begin{equation}
\mathbb{E}\left[\textbf{tr}\big(1_{[E-\epsilon,E+\epsilon]}(\tilde{H_\omega}(\Lambda_L))\big)\right] \leq C \epsilon L^{D+2}.
\end{equation}
Now, we can apply the results of Section~\ref{sec1} and prove a Minami estimates.

\begin{theo}
Let $L\in\N$, $E\in\R$, $\kappa>1$ and $\epsilon=L^{-\kappa}$. There exists $C:=C_{\kappa,K,\lambda}>0$ such that we have the following Minami estimate : 
\begin{equation}
\begin{split}
\mathbb{P}\left( \textbf{tr}\left(\textbf{1}_{[E-\epsilon,E+\epsilon]}(H_\omega(\Lambda_L))\right)\geq 2\right) 
 &\leq \mathbb{P}\left( \textbf{tr}\left(\textbf{1}_{[E-2\epsilon,E+2\epsilon]}(\tilde{H_\omega}(\Lambda_L))\right)\geq  2\right)\\
&\leq C (\epsilon L)^2 (\log L)^{\eta\beta}.
\end{split}
\end{equation}
\end{theo}

In order to prove decorrelation estimates for distinct eigenvalues, we use the following hypotheses : 

\begin{equation*}
\textbf{(H1)} : (d_n)_n \subset \R_+^\Z \text{ or } (d_n)_n \subset \R_-^\Z.
\end{equation*}
\begin{center}
\textbf{(H2)} : $(d_n)_n$ has compact support.
\end{center}
\begin{center}
\textbf{(H2')} : for all $\theta\in[0,2\pi),\sum_{n} d_n e^{in\theta} \neq 0$.
\end{center}

We will use either (H1) and (H2) or (H1) and (H2'). Assumption (H1) proves that the eigenvalues are monotonous functions of any random variable. From now on we assume that (H1) holds, and without loss of generality that $(d_n)_n \subset \R_+^\Z$. Assumptions (H2) and (H2') will be used to show that if the gradient of two eigenvalues are almost co-linear, we obtain the same conditions on the eigenvectors as in Theorem~\ref{probcoli}.

 Before proving decorrelation estimates we show some properties of the random variables $(z_n)_n$ defined in \eqref{defzn} that will be needed in the proof.

\begin{lem}\label{exp-deacr}
Define $\gamma:=d(0)^2-d(-1)d(1)>0$. There exists $N\in \N$ such that for all $n\in\Z$, if $|n|\geq N$ then $d(n)\leq d(0)e^{-\lambda|n|/2}$ and $\dfrac{e^{-\lambda N/2}}{1-e^{-\lambda N/2}}<\dfrac{1}{8\gamma d(0)^2}.$
\end{lem}

We recall the Gershgorin theorem \cite{1965}:
\begin{theo}\label{gersh}
Let $A=[a_{i,j}]_{(i,j)}\in M_n(\R)$ be a matrix such that 
\begin{equation*}
\forall i\in\{1,\dots,n\}, |a_{i,i}|>\sum_{j\neq i} |a_{i,j}|=:R_i.
\end{equation*}
Then $sp(A)\subset \bigcup_i D(a_{i,i},R_i)$.
\end{theo}

\begin{prop}\label{condm}
There exists $C>0$ (independent of L) such that for $J$ finite subset of $\Z$ such that for all $(i,j)\in J^2$, $i\neq j$, $|i-j|\geq N$, where N is defined in Lemma~\ref{exp-deacr}, then the collection $(z_j)_{j\in J\cup J+1}$ has a density bounded by $(C\|\rho\|_\infty)^{2\sharp J}$.
\end{prop}

\begin{proof}
Define $b:=d\textbf{1}_{\{-S,\dots,S\}}$, then $z_i=\sum_{n\in\Z}b(n)\omega_{n+i}$. Now, consider the application
\begin{displaymath}
(\omega_i)_{i\in J\cup J+1}\rightarrow (z_i)_{i\in J\cup J+1}.
\end{displaymath}
 The Jacobian determinant of the application is the determinant of the matrix $A=[a_{ij}]_{(i,j)\in (J\cup J+1)^2}$ where $a_{ij}=b(i-j)$. As the definition of the matrix $A$ suggest, $A$ can be rewritten by block of size two, and the blocks on the diagonal are all equal to the matrix
$\left(\begin{matrix}
d(0)& d(1)\\ d(-1) & d(0)
\end{matrix}\right)$. Its largest coefficient is equal to $d(0)$ so that the largest coefficient of its inverse is equal to $\dfrac{d(0)}{\gamma}$. Let $A=D+R=D(I+D^{-1}R)$ be this block decomposition.

 First, we remark that for $i\in J$, using Lemma~\ref{exp-deacr},
\begin{align*}
\sum_{j\in J,j\neq i} a_{ij}&\leq a(0)\sum_{j\in J,j\neq i}e^{-\lambda'|i-j|}\\
&\leq 2d(0) \sum_{k=1}^{\left\lfloor\frac{\sharp J}{2}\right\rfloor+1}e^{-k\lambda'N}\\
&\leq 2d(0)e^{-\lambda' N}\dfrac{1-e^{-(\lfloor\frac{\sharp J}{2}\rfloor+1)\lambda'N}}{1-e^{-\lambda' N}}\\
&\leq 2d(0) \dfrac{e^{-\lambda' N}}{1-e^{-\lambda' N}} \leq \dfrac{\gamma}{4d(0)}
\end{align*}

by assumption. Using Theorem~\ref{gersh}, if $\mu$ is an eigenvalue of $I+D^{-1}R$, then $\mu\geq \dfrac{1}{2}$. So the Jacobian is larger than $\left(\dfrac{\gamma}{2}\right)^{\sharp J}$. Now, $(\omega_n)_n$ are independent random variables with a common compactly supported, bounded density $\rho$. This completes the proof of Lemma~\ref{condm}.

\end{proof}

To prove decorrelation estimates we will use the same reasoning as in Section 3. Let $E(\omega)$ be a simple eigenvalue of $H_\omega(\Lambda_l)$ and $u$ a normalized eigenvalue. Define $b:=d\textbf{1}_{\{-S,\dots,S\}}$.
As $z_n=\sum_i b_i \omega_{i+n}$, using (H1) we obtain :

\begin{equation}\label{alloyder}
\partial_{\omega_m} E = \sum_i b_i u^2(m-i)=\sum_{i=1}^l b_{m-i} u^2(i)\geq 0.
\end{equation}
where $u^2(m)=0$ if $m\notin\{1,\dots, l\}$. Thus $\|\nabla_\omega E_j\| = \sum_i b_i=:A\leq \|(b_n)_n\|_1$.
We prove the following lemma
\begin{lem}\label{lemgradalloycoll}
Suppose
\begin{equation}\label{gradalloycoll}
\left\|\nabla_\omega E_j - \nabla_\omega E_j\right\|_1\leq e^{-l^\beta}.
\end{equation}
Then,
\begin{equation}
\forall i\in\{1,\dots,L\} , |u^2(i)- v^2(i)|\leq e^{-l^\beta/2}.
\end{equation}

\end{lem}
\begin{proof}

 For $i\in\{1,\dots,L\}$ define $X(i):=u^2(i)- v^2(i)$ and define \\
 $B:=[b_{ij}]_{(i,j)\in \{1-S,\dots,l+S\}\times \{1,\dots, l\} }$ by $b_{ij}=b(i-j)$ where $b:=d\textbf{1}_{\{-S,\dots,S\}}$. Then, \eqref{gradalloycoll} is equivalent to $\|BX\|\leq e^{-l^\beta}$.

Define $n_0:=\max \{n\in\N, b_n b_{-n}\neq 0\}$. Then $n_0\leq S$ and,  without loss of generality, we can assume that $b(n_0)\neq 0$. We divide the proof into two parts, according to whether (H2) or (H2') holds.

Assume (H2) holds. Let $\tilde{B}$ be the square sub-matrix of $B$ defined by
\begin{equation}
\tilde{B}:=[b_{ij}]_{(i,j)\in\{1-n_0,\dots,1-n_0+l\}\times \{1,\dots, l\}}.
\end{equation} Then, $\tilde{B}$ is a lower triangular matrix with $b_{n_0}$ on the diagonal. Thus, there exists $K>0$ (only depending on the sequence $(d_n)_n$) such that $\|\tilde{B}^{-1}\|\leq K$.

Now assume (H2') holds instead of (H2). As the sequence $(a_n)_n$ is exponentially decaying, for L large enough the sequence $(b_n)_n$ also satisfy (H2'). Define $\tilde{B}:=[b_{ij}]_{(i,j)\in  \{1-S,\dots, l+S\}^2}$. Then $B$ is a sub matrix of $\tilde{B}$ and $\tilde{B}$ is a circulant matrix. Define $\tilde{X}\in \R^{l+2S}$ by $\tilde{X}(i)=X(i)$ if $i\in\{1,\dots, l\}$ and zero elsewhere. Then, \eqref{gradalloycoll} is proved if we prove that $\|\tilde{B}\tilde{X}\|\leq e^{-l^\beta}$. Using (H2') and \cite[Proposition 5]{V10}, there exists $K>0$ (only depending on the sequence $(d_n)_n$) such that $\|\tilde{B}^{-1}\|\leq K$.
This complete the proof of Lemma~\ref{lemgradalloycoll}.

\end{proof}

Thus, we are now able to give a proof of the

\begin{lem}\label{probcolialloy}
Let $E,E'\in\sigma^2$ and $\beta>1/2$. Let $\mathbb{P}$ denote the probability that there exist $E_j(\omega)$ and $E_k(\omega)$, simple eigenvalues of $(-\Delta_a+V_\omega)_{\Lambda_l}$ and $(-\Delta_a+V_\omega)_{\Lambda_l}$ in $[E-e^{-l^\beta},E+e^{-l^\beta}]$ and $[E-e^{-l^\beta},E+e^{-l^\beta}]$ such that
\begin{equation}
\left\|\dfrac{\nabla_\omega\big(E_j(\omega))}{\|\nabla_\omega\big(E_j(\omega))\|_1}\pm\dfrac{\nabla_\omega\big(E_k(\omega))}{\|\nabla_\omega\big(E_k(\omega))\|_1}\right\|_1\leq e^{-l^\beta}.
\end{equation}
Then there exists $c>0$ such that
\begin{equation}
\mathbb{P}\leq e^{-c l^{2\beta}}
\end{equation}
\end{lem} 

\begin{proof}
The proof is the same as the proof of Lemma~\ref{probcoli} except that we obtain the conditions \eqref{cond1}, \eqref{cond2}, \eqref{cond3} on the random variables $(z_n)_n$ instead of conditions on $(\omega_n)_n$. More precisely we obtain $cl^\beta$ conditions of two type :

\begin{equation}
\left|\left(\mu_E + z_{n_+}-\lambda_E \dfrac{c}{b}\right)\left(\mu_E + z_{n_-}-\lambda_E \dfrac{c}{b}\right)-\lambda_E^2\left(1-\dfrac{a}{b}\right)^2\right|\leq e^{-l^\beta/2}
\end{equation}
with $n_+-n_-=1$ and 
\begin{equation}
\left|\mu_E + z_n-\lambda_E \dfrac{c}{a}\right|\leq e^{-l^\beta/2}.
\end{equation}
Now take $c'l^\beta$ conditions among them such that the random variables satisfies the Lemma~\ref{condm}. Then, for $L$ sufficiently large, the probability of satisfying these conditions is smaller than :

\begin{equation}
 2^l C^l \|\rho\|_\infty^{l} e^{-cl^{2\beta}} \leq e^{-c'l^{2^\beta}}\leq L^{-2}.
\end{equation}
Thus ,
\begin{equation}
\mathbb{P}\leq L^{-2}.
\end{equation}
This conclude the proof of Lemma~\ref{probcolialloy}
\end{proof}
Now, using \eqref{alloyder}, Lemma~\ref{square} is still valid. Thus, we still obtain \eqref{proba2} and Theorem~\ref{deco} for $\tilde{H}_\omega(\Lambda_L)$. Thus, using \eqref{cut}, we obtain the following theorem.

\begin{theo}\label{deco2}
Let $\beta\in(1/2,1)$. For $\alpha\in (0,1) $ and $E\neq E'$ then for any $k>1$, there exists $C>0$ such that for $L$ sufficiently large and $kL^\alpha\leq l\leq L^\alpha/k$ we have.
\begin{displaymath}
\mathbb{P}\left(
\begin{aligned}
 tr\textbf{1}_{[E-L^{-1},E+L^{-1}]}(H_\omega(\Lambda_l))\neq 0,\\
 tr\textbf{1}_{[E'-L^{-1},E'+L^{-1}]}(H_\omega(\Lambda_l))\neq 0
\end{aligned}
\right)\leq C\left(\dfrac{l^2}{L^{4/3}}\right)e^{(\log L)^\beta}
\end{displaymath}
\end{theo}

It is not clear how to prove decorrelation estimates without assumption (H1). Indeed, in Lemma~\ref{square}, we used the monotonicity of eigenvalues, seen as functions of one random variable, the others being fixed. This prevent any pair of eigenvalues of turning over a value when one of the random variable increases. There might be other ways to prevent this behaviour.

\subsection{Decorrelation estimates for the hopping model}

In this section we consider the random hopping model defined in \eqref{fullophopp} and prove Theorem~\ref{decohopp}.
The spectrum of $H_\omega(\Lambda_L)$ is symmetric with respect to the origin. Indeed, if $\phi\in\ell^2(\Lambda_L)$ is an eigenvector of $H_\omega(\Lambda_L)$ associated to the eigenvalue $E$, then $\psi$ defined by $\psi(n)=(-1)^n\phi(n)$ is an eigenvector associated to $-E$. Thus, there is no chance for decorrelation estimates for opposite eigenvalues and the study of decorrelation estimates for eigenvalue of opposite sign is equivalent to the study of decorrelation estimates for positive eigenvalues.

As in Subsection~\ref{subsecdec1}, Theorem~\ref{decohopp} is a consequence of Theorem~\ref{mina-D}, (Loc)(I) and the following theorem :

\begin{theo}\label{decohopplog}
For any $E>E'>0$ there exists $C>0$, such that for any  $\xi'\in(0,\xi)$, for  $L$ large enough and $l=(\log L)^{1/\xi'}$ we have
\begin{displaymath}
\mathbb{P}_1:=\mathbb{P}\left(
\begin{aligned}
 tr\textbf{1}_{[E-L^{-1},E+L^{-1}]}(H_\omega(\Lambda_l))=1,\\
 tr\textbf{1}_{[E'-L^{-1},E'+L^{-1}]}(H_\omega(\Lambda_l))=1
\end{aligned}
\right)\leq C\left(\dfrac{l^2}{L^{4/3}}\right) e^{l^\beta}.
\end{displaymath}
\end{theo}

We now prove Theorem~\ref{decohopplog}. Fix $E>E'>0$ and $l:=(\log L)^{1/\xi'}$. Define $J_L:=[E-L^{-1},E+L^{-1}]$, $J_L':=[E'-L^{-1},E'+L^{-1}]$ and $\epsilon:=L^{-\alpha}$ with $\alpha\in(1/2,1)$. Using Theorem~\ref{minasmall}, the event

\begin{center}
$\Omega_0(\epsilon)= \left\{ \omega;
\begin{aligned}
\sigma(H_\omega(\Lambda_l))\cap J_L&= \{E(\omega)\} \\
\sigma(H_\omega(\Lambda_l))\cap (-\epsilon,&+\epsilon)= \{E(\omega)\} \\
\sigma(H_\omega(\Lambda_l))\cap J_L'&= \{E'(\omega)\} \\
\sigma(H_\omega(\Lambda_l))\cap(-\epsilon,&+\epsilon)= \{E'(\omega)\}
\end{aligned} 
\right \}$
\end{center} 
 satisfies :
 \begin{equation}
 \mathbb{P}_1 \leq \mathbb{P}\left(\Omega_0(\epsilon)\right)+L^{-2\alpha}(\log L)^C,
\end{equation}
 for some $C>2$.
Thus, it remains to evaluate $\mathbb{P}\left(\Omega_0(\epsilon)\right)$. 

As in Subsection~\ref{subsecdec1}, we study the gradients of eigenvalues.  Let $E_j(a_\omega)$, $E_k(a_\omega)$ be two simple, non-negative eigenvalues of $H_\omega(\Lambda_l)$ and $\phi$, $\psi$  normalized eigenvectors. Then, we compute

\begin{equation}
\partial_{a_\omega(k)}E_j= 2\phi(k)\phi(k-1).
\end{equation}
The main difference with the alloy-type model is that the derivative is not necessarily positive, which was central in the proof of Lemma~\ref{square}.  
Multiplying the eigenvalue equation at point $k$ by $\phi(k)$ we obtain : 
\begin{equation}\label{polareq}
a_\omega(k+1)\partial_{a_\omega(k+1)}E_j+a_\omega(k)\partial_{a_\omega(k)}E_j= E_j\phi^2(k)
\end{equation}
Equation \eqref{polareq} shows that if we take the polar coordinates of the couple of variables $(a_\omega(k),a_\omega(k+1))$, the eigenvalues increase according to the radius. As in the proof of Lemma~\ref{square}, we need two independent growths. Thus, we will show that there is a small probability that all of the determinant $J_{i,j}:=\begin{vmatrix}
\phi^2(i)& \psi^2(i)\\
\phi^2(j)& \psi^2(j)
\end{vmatrix}$ with $|i-j|\geq 2$ (instead of $i\neq j$, as in Lemma~\ref{grad->jac}) are exponentially small. In this purpose we prove the two following lemmas, which are generalizations of Lemma~\ref{grad->jac}.

\begin{lem}\label{grad->jac2}
Let $(u,v)\in(\R_+)^{n}$ such that $\|u\|_1=\|v\|_1=1$. Fix $j_0\in \{1,\dots n\}$ such that $\max_j |u(j)|=|u(j_0)|\geq \dfrac{1}{n}$. Define $\mathcal{J}:=\{j_0-1,j_0,j_0+1\}$, $\mathcal{K}:=\{1,\dots n\}-\mathcal{J}$ and $\lambda:=\dfrac{v(j_0)}{u(j_0)}$. Suppose there exists $C>1$ such that $\left\{\|\restriction{u}{\mathcal{K}}\|_1,\|\restriction{v}{\mathcal{K}}\|_1 \right\}\in\left(\dfrac{1}{C},1\right)$ and suppose there exists $\epsilon\in\left(0,\dfrac{1}{2n^2C}\right)$ such that
\begin{equation}\label{smalljac}
\max_{k\in \mathcal{K}} \left | \begin{pmatrix} u_{j_0} & u_k\\ v_{j_0} & v_k \end{pmatrix} \right |\leq \epsilon.
\end{equation}
Then $\lambda\in\left[\dfrac{1}{2C},2C\right]$ and $\|\restriction{v}{\mathcal{K}}-\lambda \restriction{u}{\mathcal{K}} \|_1\leq n^2\epsilon$.
\end{lem}

\begin{proof}
Take $\lambda:=\dfrac{v_{j_0}}{u_{j_0}}>0$ and for $k\in \mathcal{K}$ define $\delta_k:=v_k-\lambda u_k$. We know from \eqref{smalljac} that $|\delta_k| \leq n\epsilon$. Thus $\|\restriction{v}{\mathcal{K}}-\lambda \restriction{u}{\mathcal{K}} \|_1\leq \sum_{k\in\mathcal{K}}|\delta_k|\leq n^2\epsilon$. Besides, $\sum_{k\in\mathcal{K}} \delta_k= \|\restriction{v}{\mathcal{K}}\|_1-\lambda \|\restriction{u}{\mathcal{K}}\|_1$ so $\lambda=\dfrac{\|\restriction{v}{\mathcal{K}}\|_1-\sum_{k\in\mathcal{K}} \delta_k}{\|\restriction{u}{\mathcal{K}}\|_1}$. Hence 
\begin{align*}
\lambda\leq \dfrac{1+n^2\epsilon}{1/C}\leq 2C,\\
\lambda\geq \dfrac{1}{C}-n^2\epsilon\geq \dfrac{1}{2C}.
\end{align*}
This complete the proof of Lemma~\ref{grad->jac2}.

\end{proof}

\begin{lem}\label{resteiglargehopp}
There exists $C>1$ such that for $E_j(\omega)$  an eigenvalue of $H_\omega(\Lambda_l)$ and $\phi$ a normalized eigenvector, if we define $u\in\R^l$ by
$u(i)=\phi^2(i)$ and if $\mathcal{K}$ and $\mathcal{J}$ are defined as in Lemma~\ref{grad->jac2}, we have $\|\restriction{u}{\mathcal{K}}\|_1\in\left(\dfrac{1}{C},1\right)$.
\end{lem}
\begin{proof}
As for all $n\in\N$, $|a_\omega(n)|\in\left[\dfrac{1}{M},M\right]$, using the exponential growth of the eigenvector as in the proof of Lemma~\ref{lem-large}, there exists $C>1$ (only depending on $M$) such that for $l$ sufficiently large $S_1:=\sum_{k\in\mathcal{K}}u(k)\geq \dfrac{u(j_0)}{C}$. As $\sum_{k\in\{1,\dots,l\}}u(k) = 1$, either $S_1\geq \dfrac{1}{2}$ or $u(j_0-1)+u(j_0)+u(j_0+1)\geq \dfrac{1}{2}$. In the latter case, as $u(j_0)=\max_j u(j)$, we have $u(j_0)\geq \dfrac{1}{6}$ thus $S_1\geq \dfrac{1}{6C}$. This complete the proof of Lemma~\ref{resteiglargehopp}.

\end{proof}

Let $j_0$, $\mathcal{J}$ and $\mathcal{K}$ be as in Lemma~\ref{grad->jac2}. We are now able to give a proof of  

\begin{lem}\label{probcoli2}
Fix $E,E'\in\R_+^*$ with $E\neq E'$ and $\beta>1/2$. Let $\mathbb{P}$ denotes the probability that there exist $E_n(\omega)$ and $E_m(\omega)$, simple eigenvalues of $H_\omega(\Lambda_l)$ in $[E-e^{-l^\beta},E+e^{-l^\beta}]$ and $[E'-e^{-l^\beta},E'+e^{-l^\beta}]$ such that if $\phi,\psi$ are normalized eigenvector associated to $E_n$ and $E_m$ we have
\begin{equation}
\left\|\lambda\restriction{\phi^2}{\mathcal{K}}-\restriction{\psi^2}{\mathcal{K}}\right\|\leq e^{-l^\beta},
\end{equation}
then there exists $c>0$ such that
\begin{equation}
\mathbb{P}\leq e^{-c l^{2\beta}}.
\end{equation}
\end{lem}

\begin{proof}
The proof follows the one of Lemma~\ref{probcoli}. There exist $\mathcal{P},\mathcal{Q}$ with $\mathcal{P}\cap\mathcal{Q}=\emptyset$ and $\mathcal{P}\cup\mathcal{Q}=\mathcal{K}$ such that
\begin{align*}
	\forall n\in\mathcal{P}, |\sqrt{\lambda}\phi_j(n)-\phi_k(n)|\leq e^{-l^\beta/2}\text{ ,}\\
	\forall n\in\mathcal{Q}, |\sqrt{\lambda}\phi_j(n)+\phi_k(n)|\leq e^{-l^\beta/2}\text{ .}
\end{align*}
Let P (respectively Q and J) be the projector on $\ell^2(\mathcal{P})$ (respectively $\ell^2(\mathcal{Q})$ and $\ell^2(\mathcal{J}))$. Define $\tilde{P}:=P+J$. Then the eigenvalues equations are
\begin{align*}
	(\Delta_a-E)(\tilde{P}\phi+Q\phi)=h_E\text{ ,}\\
	(\Delta_a-E')(\tilde{P}\phi-Q\phi)+H_\omega(E')\left(\dfrac{J\psi}{\sqrt{\lambda}}-J\phi\right)=\dfrac{h_{E'}}{\sqrt{\lambda}}\text{ .}
\end{align*}
where $|h_E|+\dfrac{|h_{E'}|}{\sqrt{\lambda}}\leq Ce^{-l^\beta}$. Combining these equations as in the proof of Lemma~\ref{probcoli}, there exists two collection of intervals $(C_k)_k$ and $(D_k)_k$ such that if we define $\mathcal{C}:=\cup_k C_k$ and $\mathcal{D}:=\cup_k D_k$, then $\mathcal{D}=\{1,\dots,l\}-\overset{\circ}{C}$ and 
\begin{align}
\left[\sum_k C_k\Delta_a C_k+(E'-E)\right]\phi=h_1-H_\omega(E')\left(\dfrac{J\psi}{\sqrt{\lambda}}-J\phi\right),\label{combeq3-1hopp}\\
\left[\sum_k D_k\Delta_a D_k+(E'+E)\right]\phi=h_2+H_\omega(E')\left(\dfrac{J\psi}{\sqrt{\lambda}}-J\phi\right).\label{combeq3-2hopp}
\end{align}
 
 The Lemma~\ref{lem-large} is still valid with the hopping model. Thus we define $I$ as in Lemma~\ref{lem-large} and we define $\tilde{\mathcal{J}}:=\mathcal{J}\cup\{j_0-2,j_0+2\}$, $\tilde{\mathcal{K}}:=\mathcal{K}-\{j_0-2,j_0+2\}$ and $\tilde{I}=\tilde{\mathcal{K}}\cap I$ . The vector $H_\omega(E')\left(\dfrac{J\psi}{\sqrt{\lambda}}-J\phi\right)$ has support in $\tilde{\mathcal{J}}$. There are at most four sets $C_k$ that are not entirely contained in $\tilde{I}$. 
 
 Fix $n\in \tilde{I}$. If $n\notin \mathcal{C}$, \eqref{combeq3-1hopp} shows that $|\phi(n)|\leq \dfrac{e^{-l^\beta/2}}{|E-E'|}$. Thus, using Lemma~\ref{lem-large}, for L large enough, $\mathcal{C}^c$ cannot have two consecutive points in $\tilde{I}$. As $\mathcal{C}^c=\mathring{\mathcal{D}}$ , if $D_k$ is included in $\tilde{I}$, $|D_k|\in\{2,3\}$; this assumption we will be used from now on.
 
 Suppose $|D_k|=3$, $D_k:=\{n_-,n,n_+\}$. Writing \eqref{combeq3-2hopp} at point $n_+$ we obtain
\begin{equation}\label{notlength2}
\left|\phi(n_+)-\dfrac{a_\omega(n_+)}{E+E'}\phi(n)\right|\leq \dfrac{e^{-l^\beta/2}}{E+E'}.
\end{equation}
As $E>0$ and $|\phi(n)|< e^{-l^\beta/2}$, for L sufficiently  large we obtain $|\phi(n_+)|< e^{-l^\beta /2}$ which is in contradiction with the definition of $I$. This implies that $|D_k|=2$.

Suppose $D_k:=\{n_-,n_+\}$. Then, writing \eqref{combeq3-2hopp} at points $n_+$ and $n_-$ we obtain the following equation :

\begin{equation}
\left\|\begin{pmatrix}
E+E' &  -a_\omega(n_+)\\
-a_\omega(n_+) & E+E' 
\end{pmatrix}
\begin{pmatrix}
\phi(n_-)\\
\phi(n_+)
\end{pmatrix}
\right\|
\leq e^{-l^\beta}.
\end{equation}

As $D_k\subset I$, 
$\left\|\begin{pmatrix}
\phi(n_-)\\
\phi(n_+)
\end{pmatrix}
\right\|
\geq e^{-l^\beta/2}.$ Thus we obtain the following condition on $\omega_{n_+}$ : 
\begin{equation}\label{condhopp1}
\left|a_\omega(n_+)^2-(E+E')^2\right|\leq e^{-l^\beta/2}.
\end{equation}
A same reasoning shows that if $C_k$ is included in $\tilde{I}$, $|C_k|=2$ and if we write $C_k:=\{m_-,m_+\}$ we obtain the following condition on $\omega_{m_+}$ : 
\begin{equation}\label{condhopp2}
\left|a_\omega(m_+)^2-(E-E')^2\right|\leq e^{-l^\beta/2}.
\end{equation}

Now suppose $C_k$ (respectively $D_k$) is one of the four interval that are partially in $\tilde{I}$. Using an equation similar to \eqref{notlength2} at the edge included in $\tilde{I}$ we show that $|C_k\cap\tilde{I}|=1$ (respectively $|D_k\cap\tilde{I}|=1$). 
Now, as a point $m_+$ of an interval $C_k$ included in $\tilde{I}$ is equal to a point $n_-$ of an interval $D_{k'}$, we obtain $|\tilde{I}-4|$ conditions on the $(a_\omega(n))_n$ of type $\eqref{condhopp1}$ or $\eqref{condhopp2}$. There is $l$ choices for $j_0$, $2^l$ choices for $\mathcal{P}$ and $\mathcal{Q}$. This and the fact that $|\tilde{I}|\geq cl^\beta$ complete the proof of Lemma~\ref{probcoli2}.
\end{proof}

Thus, using Lemma~\ref{grad->jac2}, when the vectors $\phi^2$ and $\psi^2$ are not almost collinear, there exists $j\in\{1,\dots, l\}$ with $|j-j_0(\omega)|\geq 2$ such that 
\begin{equation}\label{bigjachopp}
\begin{vmatrix}
\phi^2(j_0)& \psi^2(j_0)\\
\phi^2(j)& \psi^2(j)
\end{vmatrix}
\geq e^{-l^{\beta}/2}.
\end{equation}
There is less than $l^2$ choices for $(j_0,j)$.

Now suppose that $(j_0,j)$ satisfies \eqref{bigjachopp} and fix $(a_\omega(n))_{n\notin\{j,j_+1,j_0,j_0+1\}}$.
From now on, we will only write the dependence of $E_n$ and $E_m$ on \\
$(a_\omega(n))_{n\in\{j,j+1,j_0,j_0+1\}}$, i.e :
\begin{displaymath}
E_p(\omega)=E_p\left(a_\omega(j),a_\omega(j+1),a_\omega(j_0),a_\omega(j_0+1)
\right)
\end{displaymath} 
for $p\in\{n,m\}$.
For $(r_1,\theta_1,r_2,\theta_2)\in \left(\R^*_+\times[0,2\pi)\right)^2$ define
\begin{equation}
F_n\left(r_1,\theta_1 ,r_2,\theta_2\right):=E_n\left(r_1 \cos\theta_1,r_1 \sin \theta_1,r_2 \cos \theta_2 ,r_2 \sin \theta_2\right).
\end{equation}
Since $M\geq|a_\omega(n)|\geq \dfrac{1}{M}$, we restrain our study to $(r_1,r_2)\in\left[\dfrac{1}{\sqrt{2}M},\sqrt{2}M\right]^2$. We compute
\begin{multline*}
\partial_{r_1} F_n :=  \cos\theta_1 \partial_{a_\omega(j)}E_n\left(r_1 \cos\theta_1,r_1 \sin \theta_1,r_2 \cos \theta_2 ,r_2 \sin \theta_2\right) \\
+\sin \theta_1 \partial_{a_\omega(j+1)}E_n\left(r_1 \cos\theta_1,r_1 \sin \theta_1,r_2 \cos \theta_2 ,r_2 \sin \theta_2\right).
\end{multline*}
Thus,
\begin{equation}\label{polargrad}
r_1 \partial_{r_1} F_n = 2 E_n \phi^2(j) \text{ and } r_2 \partial_{r_2} F_n = 2 E_n \phi^2(j_0).
\end{equation}
Besides,
\begin{equation}
\max_{(p,q)\in\{1,2\}^2}\left|\partial{r_p}\partial{r_q}F_n \right|\leq 4 \max_{k,k'} \left|\partial_{a_\omega(k)}\partial_{a_\omega(k')}E_n(\omega)\right|
\leq 4\|Hess_{a_\omega} E_n(\omega)\|.
\end{equation}
Now fix $(\theta_1,\theta_2)$. Let $J(r_1,r_2)$ be the Jacobian of the mapping \\$(r_1,r_2)\rightarrow (F_n(r_1,r_2),F_m(r_1,r_2))$ : 
\begin{equation}\label{defjachopp}
J_{j,j_0}(r_1,r_2)=\left \vert \begin{pmatrix} \partial_{r_1} F_n & \partial_{r_2} F_n\\ \partial_{r_1} F_m &\partial_{r_2} F_m
\end{pmatrix} \right \vert .
\end{equation} 
 Define $\lambda:= E'E e^{-l^\beta/2}$ and define
\begin{displaymath}
\Omega^{j,j_0}_{0,\beta}(\epsilon)= \Omega_0(\epsilon)\cap \left \{ a_\omega ;|J_{j,j_0}(r_1,r_2)|\geq \lambda \right\}.
\end{displaymath} 
Using Lemma~\ref{probcoli2}, we have
\begin{equation}
\mathbb{P}\left(\Omega_0(\epsilon)\right)\ \leq e^{-l^{2\beta}}+ l \sum_j \mathbb{P}\left(\Omega^{j,j_0}_{0,\beta}(\epsilon)\right).
\end{equation}
It remains to evaluate $\mathbb{P}\left(\Omega^{j,j_0}_{0,\beta}(\epsilon)\right)$. Lemma~\ref{Hessien} is still valid for the hopping model. Thus, when $(r_1,r_2)\in\Omega^{j,j_0}_{0,\beta}(\epsilon)$, $\|Hess_r F_n\|+  \|Hess_r F_m\| \leq \dfrac{C}{\epsilon}$
Therefore, we obtain the following lemma which corresponds to Lemma~\ref{square} and is proven in the same way.
\begin{lem}\label{squarehopp}
Pick $\epsilon= L^{-\alpha}$. For any $(a_\omega(n))_{n\notin\{j,j_+1,j_0,j_0+1\}}$, for any $(\theta_1,\theta_2)$, if there exists $(r_1^0,r_2^0)\in \R_+^2$ such that $a_\omega\in\Omega^{\gamma,\gamma'}_{0,\beta}(\epsilon)$ then for $(r_1,r_2) \in \R_+^2$ such that $|(r_1,r_2)-(r_1^0,r_2^0)|_\infty\leq \epsilon$ one has $(F_n(a_\omega),F_m(a_\omega))\in J_L\times J_L'\Longleftrightarrow |(r_1,r_2)-(r_1^0,r_2^0)|_\infty\leq L^{-1}\lambda^{-2}$. 
\end{lem}

Thus, coming back to random variables and using independence, we obtain the following equation, which is proven in the same way that \eqref{proba2}.

\begin{equation}\label{proba2hopp}
\mathbb{P}(\Omega^{\gamma,\gamma'}_{0,\beta}(\epsilon))\leq CL^{\alpha-2}\lambda^{-4}.
\end{equation} 
Optimization in $\alpha$ yields $\alpha=\dfrac{4}{3}$. This complete the proofs of Theorem~\ref{decohopplog} and Theorem~\ref{decohopp}.

\section{Spectral statistics}

Since we study models where we have (IAD) and (W), and since we have proved Theorem~\ref{mina} and (D), we can follow the strategy developed in \cite{2010arXiv1011.1832G} to study the spectral statistics of these models. 

In order to use the same notation as in \cite{2010arXiv1011.1832G}, in this section we will use the following Wegner and Minami estimates :

\textbf{(W) :} There exists $C>0$, $m\geq 1$, such that for $J\subset \mathcal{I}$ and $\Lambda\subset\N$
\begin{equation}
\mathbb{P}\Big[tr \left(\textbf{1}_J(H_\omega(\Lambda)) \right)\geq 1\Big ]\leq C |J||\Lambda|^m.
\end{equation}

\textbf{(M) :} There exists $C>0$, $\rho> 0$  such that for $J\subset \mathcal{I}$ and $\Lambda\subset\N$
\begin{equation}
\mathbb{E}\Big[tr \left(\textbf{1}_J(H_\omega(\Lambda)) \right)[tr \left(\textbf{1}_I(H_\omega(\Lambda))-1 )\right]\geq 1\Big ]\leq C (|J||\Lambda|)^{1+\rho}.
\end{equation}

The difference between assumption (W) and (M) in the present article and the ones in \cite{2010arXiv1011.1832G} is the exponent $m$ in (W). The results in \cite{2010arXiv1011.1832G} were proved with a Wegner estimate linear in volume, i.e $m=1$. As (M) is a consequence of (W), there is no modification in (M).
The next lemma corresponds to \cite[Lemma 3.1]{2010arXiv1011.1832G}, has a proof similar to that of Lemma~\ref{loc} and will be used to described the eigenvalues of $H_\omega(\Lambda)$ in terms of eigenvalues of $H_\omega$ restricted to smaller cubes. We use (Loc)(I) if theses cubes are polynomially large and (Loc)(II) if they are logarithmically large so that the probability of the complement of the set given in the assumption is polynomially small.

\begin{lem}
Assume (IAD), (W), (M) and (Loc). Consider scales $l',l$ so that $(RLog|\Lambda_L|)^{1/\xi}\leq l' \ll l \ll L$ for $R$ large enough, and, for some given $\gamma\in \Lambda_L$, consider a box $\Lambda_l(\gamma)$ such that $\Lambda_{l-l'}(\gamma) \subset \Lambda_L$. Let $\mathcal{W}_{\Lambda_L}$ be the set $\mathcal{U}_{\Lambda_L}$ or $\mathcal{V}_{\Lambda_L}$, defined in (Loc). For L large enough, we have : 
\begin{enumerate}
\item for any $\omega\in\mathcal{W}_{\Lambda_L}$, if $E(\omega)$ is an eigenvalue of $H_\omega(\Lambda_L)$  with a center of localization in $\Lambda_{l-l'}(\gamma)$, then $H_\omega(\Lambda_L\cup \Lambda_l(\gamma))$ has an eigenvalue in a $[E(\omega)-e^{-(l')^\xi/2},E(\omega)+e^{-(l')^\xi/2}]$; moreover, if $\omega\in\mathcal{W}_{\Lambda_l(\gamma)}$, the corresponding eigenfunction is localized in sense of \eqref{expdec}.
\item
Assume now additionally that $\Lambda_l(\gamma)\subset \Lambda_L$. Then, conversely, for any $\omega\in \mathcal{W}_{\Lambda_l(\gamma)}$, if $E'(\omega)$ is an eigenvalue of $H_\omega(\Lambda_l(\gamma))$ with centre of localization  in $\Lambda_{l-l'}(\gamma)$, then $H_\omega(\Lambda_L)$ has an eigenvalue in $[E'(\omega)-e^{-(l')^\xi/2},E'(\omega)+e^{-(l')^\xi/2}]$; moreover, if $\omega\in\mathcal{W}_{\Lambda_L}$, the corresponding eigenfunction is localized in sense of \eqref{expdec}.
\end{enumerate}

\end{lem}

Now pick $\tilde{\rho}\in\left[0,\dfrac{\rho}{1+\rho}\right[$ where $\rho$ is defined in (M) and pick $E_0\in I$ such that 
\begin{equation}\label{IDS1}
\forall a>b, \exists C>0, \exists \epsilon_0, \forall 0<\epsilon<\epsilon_0, |N(E_0+a\epsilon)-N(E_0-b\epsilon)|\geq C\epsilon^{1+\tilde{\rho}}.
\end{equation}
Define $\alpha_{\rho,\rho'}:=(1+\tilde{\rho})\left(1-\dfrac{1-\gamma}{\rho\gamma+1}\right)<1$ where $\gamma:=\dfrac{\rho}{1+\rho}$.\\
We will use the two following decomposition theorems. In the first we control all eigenvalues in intervals of size $\frac{1}{|\Lambda|^\alpha}$, in the second, we control most eigenvalues in intervals of size $(\log |\Lambda|)^{-1/\xi}$. 

\paragraph{Controlling all eigenvalues}
\text{ }\\
The following theorem correspond to Theorem 1.1 in \cite{2010arXiv1011.1832G}.
\begin{theo}\label{alleig}
Assume $E_0$ is such that \eqref{IDS1} holds and pick $\alpha\in(\alpha_{\rho,\rho'},1)$
Pick $I_\Lambda$ centered at $E_0$ such that $N(I_\Lambda)\asymp |\Lambda_L|^{-\alpha}$. There exists $\beta>0$ and $\beta'\in(0,\beta)$ small so that $1+\beta\rho<\alpha\dfrac{1+\rho}{1+\tilde{\rho}}$ and, for $l\asymp L^\beta$ and $l'\asymp L^{\beta'}$, there exists a decomposition of $\Lambda=\Lambda_L$ into disjoint cubes of the form $\Lambda_l(\gamma_j):= \gamma_j+[0,l]$ satisfying :
\begin{enumerate}
\item $\cup_j \Lambda_l(\gamma_j)\subset \Lambda$,
\item if $j\neq k,\text{dist}(\Lambda_l(\gamma_j),\Lambda_l(\gamma_k))\geq l'$
\item $\text{dist}(\Lambda_l(\gamma_j),\partial \Lambda)\geq l'$
\item $\left|\Lambda \backslash \cup_j \Lambda_l(\gamma_j)\right|\lesssim |\Lambda| l'/l$,
\end{enumerate}
and such that, for L sufficiently large, there exists a set of configuration $\mathcal{Z}_\Lambda$ such that 
\begin{enumerate}
\item $\mathbb{P}(\mathcal{Z}_\Lambda)\geq 1-|\Lambda|^{-(\alpha-\alpha_{\rho,\tilde{\rho}})}$
\item for $\omega\in\mathcal{Z}_\Lambda$, each centre of localization associated to $H_\omega(\Lambda)$ belong to some $\Lambda_l(\gamma_j)$ and each box $\Lambda_l(\gamma_j)$ satisfies 
	\begin{enumerate}
	\item the Hamiltonian $H_\omega(\Lambda_l(\gamma_j))$ has at most one eigenvalue in $I_\Lambda$, \\say $E_j(\omega,\Lambda(\gamma_j))$;
	\item $\Lambda_l(\gamma_j)$ contains at most one centre of localization, say $x_{k_j}(\omega,\Lambda)$, of an eigenvalue of $H_\omega(\Lambda)$ in $I_\Lambda$, say $E_{k_j}(\omega,\Lambda)$;
	\item $\Lambda_l(\gamma_j)$ contains a centre $x_{k_j}(\omega,\Lambda)$ if and only if $\sigma(H_\omega(\Lambda_l(\gamma_j))\cap I_\Lambda \neq \emptyset$; in which case, one has
	\begin{equation}\label{stat1}
	|E_{k_j}(\omega,\Lambda)-E_j(\omega,\Lambda_l(\gamma_j))|\leq e^{-(l')^\xi}\text{ and dist}(x_{k_j}(\omega,\Lambda),\Lambda \backslash \Lambda_l(\gamma_j))\geq l'.
	\end{equation}
	\end{enumerate}
\end{enumerate}
In particular if $\omega\in\mathcal{Z}_\Lambda$, all eigenvalues of $H_\omega(\Lambda)$ are described by \eqref{stat1}.
\end{theo}

\begin{proof}
We follow the proof of Theorem 1.1 in \cite{2010arXiv1011.1832G}, taking into account the different Wegner estimate.
Pick $\beta'=0+$ and $\beta>\beta'$ to be chosen later, set $l'=L^{\beta'}$ and $l$ so that
$(l+l')k+l'=L$, where $k=\lfloor L^{1-\beta}\rfloor$. Now pick boxes of size $l$ in $\Lambda_L$ satisfying the conditions of the Theorem~\ref{alleig}. We will use (Loc)(I) in this proof.
The probability that one of the cube $\Lambda_l(\gamma_j)$ does not satisfy (Loc)(I) is bounded by $|\Lambda_L|^{1-\beta-p\beta}$. The probability that $\Lambda_L$ does not satisfy (Loc)(I) is bounded by $|\Lambda_L|^{-p}$. So for $p$ large enough, up to a probability negligible with respect to $|\Lambda_L|^{-(\alpha-\alpha_{d,\rho,\tilde{\rho}})}$, all boxes satisfy (Loc)(I), which we will assume from now on.

Let $S_{l,L}$ be the set of boxes $\Lambda_{l-l'}(\gamma_j)\subset \Lambda_L$ containing two centres of localization of $H_\omega(\Lambda_L)$. It follows from the Minami estimates and Lemma~\ref{loc2} that
\begin{equation}
\mathbb{P}(\sharp \mathcal{S}_{l,L}\geq 1)\leq |\Lambda_L|^{1-\beta}(|\Lambda|^\beta|I_\Lambda|)^{1+\rho}\leq |\Lambda_L|^{1+\beta\rho}N(I_\Lambda)^{\frac{1+\rho}{1+\tilde{\rho}}}\leq |\Lambda_L|^{1+\beta\rho-\alpha \frac{\gamma^{-1}}{1+\tilde{\rho}}}
\end{equation}
where $\gamma:=\dfrac{1}{1+\rho}\in(0,1)$.\\
Define $\Upsilon=\Lambda_L\backslash\cup_j \Lambda_{l-l'}(\gamma_j)$ and consider a partition $\Upsilon = \cup_{k=1}^{M} \Lambda_{l'}(x_k)$ into boxes of size l', with $M\asymp L l^{-1}$. We can enlarge each box of  $\frac{1}{10}l'$ except for sides that coincide with the boundary of $\Lambda_L$, so that if a centre of localization is in $\Upsilon$ we can apply Lemma~\ref{loc2}. Hence, using (W) :
\begin{align*}
\mathbb{P}&(H_\omega(\Lambda_L) \text{ has a centre of localization in }\Upsilon)\\
&\lesssim \sum_k^M \mathbb{P}(H_\omega(\Lambda_L) \text{ has a centre of localization in }\Lambda_{l'}(x_k))\\
&\lesssim \sum_k^M \mathbb{P}(H_\omega(\Lambda_L) \text{ has a centre of localization in }\Lambda_{\frac{11}{10}l'-\frac{1}{10}l'}(x_k))\\
&\lesssim M|\Lambda_{\frac{11}{10}l'}|^{m}|I_\Lambda|\lesssim  L l'^{m}l^{-1}N(I_\Lambda)^{\frac{1}{1+\tilde{\rho}}}\lesssim |\Lambda_L|^{1-\beta+m\beta'-\frac{\alpha}{1+\tilde{\rho}}}.
\end{align*}
Now take
\begin{center}
$\alpha>(1+\tilde{\rho})\max\left((1+\beta\rho)\gamma,1-\beta+m\beta')\right)$.
\end{center} 

If we optimize in $\beta$ we have :
\begin{equation}
\beta=\dfrac{1-\gamma+\beta'm}{\rho\gamma+1}~\text{   and   }~ \alpha>(1+\tilde{\rho})\left(1-\dfrac{1-\gamma-\rho\gamma m\beta'}{\rho\gamma+1}\right)
\end{equation}
Thus, taking $\alpha_{\rho,\tilde{\rho}}:=(1+\tilde{\rho})\left(1-\dfrac{1-\gamma}{\rho\gamma+1}\right)=(1+\tilde{\rho})\left(\dfrac{1+\rho }{\rho+\gamma^{-1}}\right)<1$ yields the result. This complete the proof of Theorem~\ref{alleig}.

\end{proof}
\paragraph{Controlling most eigenvalues}
\text{ }\\
The following theorem correspond to Theorem 1.2 in \cite{2010arXiv1011.1832G}.
\begin{defi}
Pick $\xi\in(0,1)$, $R>1$ large enough and $\rho'\in(0,\rho)$ where $\rho$ is defined in (M). For a cube $\Lambda$, consider an interval $I_\Lambda=[a_\Lambda,b_\Lambda]\subset I$. Set $l'_\Lambda=(R\log|\Lambda|)^{1/\xi}$, we say that the sequence $(I_\Lambda)_\Lambda$ is $(\xi,R,\rho')$-admissible if, for any $\Lambda$, one has
\begin{equation}
|\Lambda|N(I_\Lambda)\geq 1,\hspace*{2em}N(I_\Lambda)|I_\Lambda|^{-(1+\rho')}\geq 1,\hspace*{2em}N(I_\Lambda)^{\frac{1}{1+\rho'}}l'_\Lambda\leq 1
\end{equation}
\end{defi}

\begin{theo}\label{mosteig}
Assume (IAD), (W), (M) and (Loc) Hold. Pick $\rho'\in\left[0,\frac{\rho}{2+\rho}\right)$ where $\rho$ is defined in (M). For any $q>0$, for $L$ sufficiently large, depending only on $\xi,R,\rho',p$; for any sequence of intervals $(I_\Lambda)_\Lambda$ that is $(\xi,R,\rho')$-admissible, and any sequence of scales $\tilde{l}_\Lambda$ so that $l'_\Lambda\ll \tilde{l}_\Lambda \ll L$ and
\begin{equation}
N(I_\Lambda)^{\frac{1}{1+\rho'}}\tilde{l}_\Lambda\underset{|\Lambda|\rightarrow \infty}\rightarrow 0
\end{equation}
there exists
\begin{enumerate}
\item a decomposition of $\Lambda=\Lambda_L$ into disjoint cubes of the form $\Lambda_{l_\Lambda}(\gamma_j)$, where $l_\Lambda=\tilde{l}_\Lambda(1+{\mathcal O}(\tilde{l_\Lambda}/|\Lambda|))=\tilde{l}_\Lambda(1+o(1))$
such that
\begin{enumerate}
\item $\cup_j \Lambda_{l_\Lambda}(\gamma_j)\subset \Lambda$,
\item if $j\neq k,\text{dist}(\Lambda_{l_\Lambda}(\gamma_j),\Lambda_{l_\Lambda}(\gamma_k)\geq l'_\Lambda$
\item $\text{dist}(\Lambda_{l_\Lambda}(\gamma_j),\partial \Lambda)\geq l'_\Lambda$
\item $\left|\Lambda \backslash \cup_j \Lambda_{l_\Lambda}(\gamma_j)\right|\lesssim \Lambda l'_\Lambda/l_\Lambda$,
\end{enumerate}

\item a set of configurations $\mathcal{Z}_\Lambda$, such that 
\begin{equation}
\mathbb{P}(\mathcal{Z}_\Lambda)\geq 1-|\Lambda|^{-q}-\exp\left(-c|I_\Lambda|^{1+\rho}|\Lambda|l_\Lambda^{\rho}\right)-\exp\left(-c|\Lambda||I_\Lambda|l'^m_\Lambda l_\Lambda^{-1}\right)
\end{equation}
so that

\item for $\omega\in\mathcal{Z}_\Lambda$, there exists at least $\dfrac{|\Lambda|}{l_\Lambda}\left(1+o\left(N(I_\Lambda)^{\frac{1}{1+\rho'}}l_\Lambda\right)\right)$ disjoint boxex $\Lambda_{l_\Lambda}(\gamma_j)$ satisfying the properties (a),(b),(c) described in (Theorem~\ref{alleig}) with $l'_\Lambda=(R\log|\Lambda|)^{1/\xi}$. Note that $N(I_\Lambda)l'_\Lambda=o(1)$.
\item
The number of eigenvalues of $H_\omega(\Lambda)$ that are not described above is bounded by
\begin{equation}
CN(I_\Lambda)|\Lambda|\left(N(I_\Lambda)^{\frac{\rho-\rho'}{1+\rho'}}l_\Lambda^{1+\rho}+N(I_\Lambda)^{-\frac{\rho'}{1+\rho'}}(l'_\Lambda)^{1+m}l_\Lambda^{-1}\right)
\end{equation}
and this number is $o\left(N(I_\Lambda)|\Lambda|\right)$ if we have 
\begin{equation}
N(I_\Lambda)^{-\frac{\rho'}{1+\rho'}}(l'_\Lambda)^{1+m} \ll l_\Lambda \ll N(I_\Lambda)^{-\frac{\rho-\rho'}{(1+\rho)(1+\rho')}}
\end{equation}
where $m$ is defined in (W).

\end{enumerate}

\end{theo}

The proof of Theorem~\ref{mosteig} is the same as the proof of \cite[Theorem 1.2]{2010arXiv1011.1832G} taking into account the different Wegner estimates, as done in the proof of Theorem~\ref{alleig}.
\subsection{The local distribution of eigenvalues}

As in \cite[Section 2]{2010arXiv1011.1832G}, we compute the distribution of unfolded eigenvalues, i.e the numbers \\
$\Big(|\Lambda_L|N\big(E_j(\omega)\big)\Big)_j$.
As the assumption (W) is not linear in the volume of the cube, it does not guarantee that the IDS is Lipschitz continuous. So we make an assumption on the regularity of the IDS. Define $R_\eta(l):=\sup\limits_{|I|\leq e^{-l^\eta}} N(I)$.  As we are in dimension one, one can prove the following regularity  estimate using techniques of \cite{bougerol1985products}:  
\begin{equation}
\textbf{(R)}: \forall\eta\in(0,1), \lim\limits_{L\to\infty} L^d R_\eta(l')=0.
\end{equation}

where $l'$ is defined either in  Theorem~\ref{alleig} or in Theorem~\ref{mosteig}. In Theorem~\ref{alleig} and Theorem~\ref{mosteig} we use two different lengths $l'$ so assumption (R) has two different meanings.

 $\bullet$ In Theorem~\ref{alleig}, we take $l'=L^{\beta'}$. Hence, (R) is true provided the IDS is Log-Hölder continuous with an exponent sufficiently large and says that the IDS of exponentially small intervals should be negligible over polynomially small quantities.

 $\bullet$ In Theorem~\ref{mosteig} , we take $l'=R(\log L)^{1/\xi}$. Hence, in this case, (R) is true provided the IDS is Hölder continuous and R is sufficiently large .

 For instance, in Appendix~\ref{ap2} we prove :
\begin{theo}\label{example}
Let $(\omega_n)_{n\in\N}$ be random variables i.i.d with a density which support is $[0,1]$. Let $
H_\omega:=-\Delta+V_\omega$ with $V_\omega(2i)= -V_\omega(2i+1)=\omega_i$,  then, the IDS is well defined and Holder continuous in $[-3,
3]$. 
\end{theo}
If the density is sufficiently regular, a Wegner estimate near the infimum  of the almost sure spectrum and polynomial in volume is known to hold (\cite{Kl95}). \\
To compute the local distribution of eigenvalues we use the following lemma : 
\begin{lem}\label{unfoldist}
Assume (W) and (Loc) hold in I a compact and assume (R). For any $\nu\in(0,1)$ and $ \log L \ll l' \ll l \ll L$, let $N(J,l,l')$ be the number of eigenvalue of $H_\omega(\Lambda_l)$ in J with localization centre in $\Lambda_{l-l'}$. Then, there exists $C>0$ and $1>\xi''>\nu$ such that, for $J\subset I$ an interval such that $|J|\geq e^{-(l')^\nu}$, one has
\begin{equation}
|\mathbb{E}(N(J,l,l'))-N(J)|\Lambda_l|\lesssim |J|l'+R_{\xi''}(l')l.
\end{equation}
\end{lem}

\begin{proof}
Let $\mathcal{U}_{\Lambda_l}$ be the set of configuration given in (Loc)(I) for some $p\in\mathbb{N}^*$ and $\xi\in(0,1)$. Let $\chi_r$ be a smooth function such that $0\leq\chi_r\leq 1$ and $\chi_r=1$ on $J$ and $\chi_r=0$ outside $J+[-r,r]$ with $r:=e^{-(l')^{\xi''}}$ and $\xi''$ to be chosen later. Define $\mathcal{Y}_{l,l'}$ the subset of $\mathcal{U}_{\Lambda_l}$ where $H_\omega(\Lambda_l)$ has all localization centres in $\Lambda_{l-l'}$.
Then, by (W) and Lemma~\ref{loc2}, we have

\begin{equation}
\left|\mathbb{E}\Big(N(J,l,l')-\text{tr}[1_J(H_\omega(\Lambda_l))]\Big)\right|  \leq\mathbb{P}\left(\mathcal{Y}_{l,l'}^c\right)
\leq |J|(l')^m.
\end{equation}
Now, define $K:=(J+[-r,r])-J$. Then, $|K|= 2r$ and
\begin{equation*}
0\leq r(x)-1_J(x)\leq 1_K(x).
\end{equation*}
Thus, using (W) we obtain
\begin{equation}
\left|\mathbb{E}\Big[\text{tr}[1_J(H_\omega(\Lambda_l))-\chi_r(H_\omega(\Lambda_l))]\Big]\right|\lesssim r|\Lambda_l|^m.
\end{equation}
In the same way, we obtain 
\begin{equation}
\left|\mathbb{E}\Big[\text{tr}[1_{\Lambda_l}1_J(H_\omega)-1_{\Lambda_l}\chi_r(H_\omega)]\Big]\right|\leq R_{\xi''}(l')l.
\end{equation}
Now, $H_\omega(\Lambda_l)$ can be extended to the operator $1_{\Lambda_l}H_\omega 1_{\Lambda_l}$ acting on $\ell^2(\Z)$. Thus, using the resolvent equation we have : 
\begin{multline}
\left(1_{\Lambda_l}H_\omega 1_{\Lambda_l}-z\right)^{-1}-\left(1_{\Lambda_l}H_\omega-z\right)^{-1} \\
=\left( 1_{\Lambda_l}H_\omega  1_{\Lambda_l}-z\right)^{-1}\left(1_{\Lambda_l}H_\omega-1_{\Lambda_l}H_\omega 1_{\Lambda_l}\right)\left(1_{\Lambda_l} H_\omega -z\right)^{-1}.
\end{multline}
Let $x\in\Z$ and define $\Psi_x:=\left(1_{\Lambda_l} H_\omega -z\right)^{-1}\delta_x$. We compute
\begin{equation*}
\left(1_{\Lambda_l}H_\omega-1_{\Lambda_l}H_\omega 1_{\Lambda_l}\right)\Psi_x=1_{\Lambda_l}H_\omega\left(1-1_{\Lambda_l}\right)\Psi_x=\Psi_x(l+1)\delta_l+\Psi_x(-l-1)\delta_{-l}.
\end{equation*}
Thus,
\begin{multline*}
\left\langle \delta_x,\left( 1_{\Lambda_l}H_\omega  1_{\Lambda_l}-z\right)^{-1}\left(1_{\Lambda_l}H_\omega-1_{\Lambda_l}H_\omega 1_{\Lambda_l}\right)\Psi_x\right\rangle\\
=\Psi_x(l+1)\langle \delta_x,\left( 1_{\Lambda_l}H_\omega  1_{\Lambda_l}-z\right)^{-1}\delta_l\rangle\\
+\Psi_x(-l-1)\langle \delta_x,\left( 1_{\Lambda_l}H_\omega  1_{\Lambda_l}-z\right)^{-1}\delta_{-l}\rangle.
\end{multline*}

Now, let $(E_j)_j$ denote the eigenvalues of $H_\omega(\Lambda_l)$ and $\phi_j$ normalized eigenvectors extended to $\ell^2(\Z)$ by zeros.
Then, we compute
\begin{equation}\label{resolv}
\langle \delta_x,\left( 1_{\Lambda_l}H_\omega  1_{\Lambda_l}-z\right)^{-1}\delta_l\rangle=\sum_j \dfrac{\phi_j(l)\phi_j(x)}{E_j-z}.
\end{equation}
The right-hand side of \eqref{resolv} vanishes if $x\notin\{1,\dots,l\}$ and if $(\omega_n)_n\in\mathcal{Y}_{l,l'}$, has its modulus smaller than $\dfrac{|\Lambda_l|}{|\mathrm{Im}z|}e^{-(l')^\xi}$. The same estimation can be done with the term \\
$\langle \delta_x,\left( 1_{\Lambda_l}H_\omega  1_{\Lambda_l}-z\right)^{-1}\delta_{-l}\rangle$. Now, we know that $\|\Psi_x\|\leq \dfrac{1}{|\mathrm{Im}z|}$. Therefore, we have 

\begin{equation}
\left|\text{tr}\left[\left(1_{\Lambda_l}H_\omega-z\right)^{-1}-\left(1_{\Lambda_l}H_\omega 1_{\Lambda_l}-z\right)^{-1}\right]\right|\leq |\Lambda_l|\dfrac{e^{-(l')^\xi}}{|\mathrm{Im}z|^2}.
\end{equation}

 Now, using the Helffer-Sjöstrand formula (cf \cite[Appendix B]{Hunziker97time-dependentscattering})  to represent \\ $\chi_r(H_\omega(\Lambda_l)) $ and $\chi_r(H_\omega)$ using theirs 
 resolvent there exists $C>0$ such that
 
\begin{displaymath}
\left|\mathbb{E}\big[\textbf{1}_{\mathcal{Y}_{l,l'}}\text{tr}[\chi_r(H_\omega(\Lambda_l))-1_{\Lambda_l}\chi_r(H_\omega)]\big]\right|\leq |\Lambda_l|r^{-C}e^{-(l')^{\xi}}
 \lesssim |\Lambda_l|e^{-(l')^{\xi}}
\end{displaymath}
if one take $\xi''\in(\nu,\xi)$. This conclude the proof of Lemma~\ref{unfoldist}.

\end{proof}

On $\mathcal{U}_{\Lambda_l}$, define the Bernoulli random variable $X=X(J,l,l')$ such that $X=1$ if and only if $H_\omega(\Lambda_l)$ has only one eigenvalue in J with centre of localization in $\Lambda_{l-l'}$. When $X=1$, let $E_j$ be this eigenvalue.
From now on we will take $J$ such that $N(J)\asymp |\Lambda_L|^{-\alpha}$. By assumption (R), $J$ satisfies $|J|\geq e^{-(l')^\nu}$.
Using (M) as in \cite[Lemma 2.1]{2010arXiv1011.1832G}, we have the following result  : 

\begin{theo}
Assume (W), (M) and (Loc) hold in I a compact and assume (R). For $\nu\in(0,1)$ and $ R\log L \ll l' \ll l \ll L$ with R sufficiently large. Then, for any interval $J\subset I$ such that $|J|\geq e^{-(l')^\nu}$, one has

\begin{equation}\label{loclev}
\left| \mathbb{P}(X=1)-N(J)l^d\right|\lesssim (|I|l)^{1+\rho}+|I|l^m+R_\eta(l')l
\end{equation}

\end{theo}

As we will use this estimate with $N(J)\asymp |\Lambda_L|^{-\alpha}$, assumption $\textbf{(R)}$ shows that the terms in the right-hand side of \eqref{loclev} are error terms.\\
Let us now state our  results on the spectral statistics for $H_\omega(\Lambda_L)$. They are similar to those found in \cite{2010arXiv1011.1832G}; the proofs will not be given as they are the same as those in \cite{2010arXiv1011.1832G}.

\subsection{The spectral statistics}

\subsubsection{Unfolded local level statistics}

The \textit{unfolded local level statistics} near $E_0$ is the point process defined by
\begin{equation}
\Xi(\xi;E_0,\omega,\Lambda)=\sum_{j\geq1} \delta_{\xi_j(E_0,\omega,\Lambda)}(\xi)
\end{equation}
 where
 \begin{equation}
 \xi_j(E_0,\omega,\Lambda)=|\Lambda|(N(E_j(\omega,\Lambda)-N(E_0))
 \end{equation}

The unfolded local level statistics is described by the following theorem which corresponds to \cite[Theorem 1.10]{2010arXiv1011.1832G}.

\begin{theo}\label{ULLS2}
Pick $\overline{\rho}\in\left[0,\dfrac{\rho}{1+\rho}\right[$ where $\rho$ is defined in (M).\\
Pick $E_0\in I$ such that 
\begin{multline}\label{notfastdecreas}
\forall \delta\in(0,1), \exists C(\delta)>0, \exists \epsilon_0, \forall 0<\epsilon<\epsilon_0,\forall a\in[-1,1],\\
 |N(E_0+(a+\delta)\epsilon)-N(E_0+a\epsilon)|\geq C(\delta)\epsilon^{1+\overline{\rho}}.
\end{multline}
Pick $\alpha\in(\alpha_{d,\rho,\tilde{\rho}},1)$. Then there $\delta>0$ such that, for any sequence of intervals $I_1=I_1^\Lambda,\dots,I_p=I_p^\Lambda$ in $|\Lambda|^{1-\alpha}.[-1,1]$, where p may depend on $\Lambda$ satisfying
\begin{equation}
\inf_{j\neq k} \text{dist}(I_j,I_k)\geq e^{-|\Lambda|^\delta},
\end{equation}
we have, for any sequences of integers $k_1=k_1^\Lambda,\dots,k_1=k_p^\Lambda$

\begin{equation}
\lim_{|\Lambda|\to\infty}\left|\mathbb{P}
\left(
\left\{\omega;
\begin{aligned}
\sharp\{j;\xi_j(\omega,\Lambda)\in I_1\}=k_1\\
\vdots\hspace*{8em} \vdots\hspace*{1em}\\
\sharp\{j;\xi_j(\omega,\Lambda)\in I_p\}=k_p
\end{aligned}
\right\}\right)-\dfrac{|I_1|^{k_1}}{k_1!}\dots\dfrac{|I_p|^{k_p}}{k_p!}\right|=0
\end{equation}
In particular, $\Xi(\xi;E_0,\omega,\Lambda)$ converges weakly to a Poisson process with intensity Lebesgue.

\end{theo}

Now, Theorem~\ref{ULLS} is a consequence of Theorem~\ref{ULLS2}. Indeed, condition \eqref{notfastdecreas} is in particular satisfied at any point $E_0$ where $N$ is differentiable with positive derivative.
\subsubsection{Asymptotic independence of the local processes}

Now, as in \cite{2010arXiv1011.1832G}, we show how the point processes associated to different energies relate one another, using (D). The following theorem corresponds to \cite[Theorem 1.11]{2010arXiv1011.1832G}.

\begin{theo}\label{AILP}
Assume (IAD), (W), (M), (Loc) and (D) hold. Pick $E_0\in \mathcal{I}$ and $E_0'\in \mathcal{I}$ such that $E_0\neq E_0'$ and (\ref{IDS1}) is satisfied at $E_0$ and $E_0'$.\\
When $|\Lambda|\rightarrow \infty$ the point processes $ \Xi(E_0,\omega,\Lambda)$ and $\Xi(E_0',\omega,\Lambda)$, converge weakly respectively to two independent Poisson processes on $\R$ with intensity the Lebesgue measure. That is, for $U_+\subset\R$ and $U_-\subset\R$ compact intervals and $\{k_+,k_-\}\in \N\times \N$, one has 
\begin{displaymath}
\mathbb{P}\left(
\begin{aligned}
\sharp\{j;\xi_j(E_0,\omega,\Lambda)\in U_+\}=k_+\\
\sharp\{j;\xi_j(E_0',\omega,\Lambda)\in U_-\}=k_-
\end{aligned}
\right)\underset{\Lambda\to\Z}{\rightarrow} \left(\dfrac{|U_+|^{k_+}}{k_+!}e^{-|U_+|}\right)\left(\dfrac{|U_-|^{k_-}}{k_-!}e^{-|U_-|}\right)
\end{displaymath}
\end{theo}
\subsubsection{Level spacing statistic near a given energy}
In this section we recall \cite[Theorem 1.5]{2010arXiv1012.0831K} which is more precise than \cite[Theorem 1.4]{2010arXiv1011.1832G}.
Pick $I_\Lambda= [ a_\Lambda ,b_\Lambda ]$ such that $|a_\Lambda|+|b_\Lambda|\rightarrow 0$.
Define $\delta N_j(\omega,\Lambda)=|\Lambda|(N(E_{j+1}(\omega,\Lambda))-N(E_{j}(\omega,\Lambda)))\geq 0$, $N(I_\Lambda,\Lambda,\omega):=\sharp\{j,E_j(\omega,\Lambda)\in I_\Lambda\}$ and define the empirical distribution of these spacing to be the random number, for $x\geq 0$ :
\begin{equation}
DSL(x;E_0+I_\Lambda,\omega,\Lambda)=\dfrac{\sharp\{j;E_j(\omega,\Lambda)\in E_0+I_\Lambda,\delta N_j(\omega,\Lambda)\geq x\}}{N(I_\Lambda,\Lambda,\omega)}
\end{equation}

We use Theorem~\ref{mosteig} and obtain the following theorem :

\begin{theo}\label{LSS}
Assume (IAD),(W), (M), and (Loc) hold. Fix $E_0\in \mathcal{I}$, such that, for some $\overline{\rho}\in[0,\rho/(1+(\rho+1))]$, there exists a neighbourhood  of $E_0$, say U, such that
Fix $(I_\Lambda)_\Lambda$ a decreasing sequence of intervals such that 
\begin{equation}
\sup_{E\in I_\Lambda} |E| \underset{\Lambda\to\Z}{\rightarrow} 0.
\end{equation}
Let us assume that 
\begin{equation}
\text{if }l'=o(L)\text{ then } \lim\limits_{\Lambda\to\infty}\dfrac{N(E_0+I_{\Lambda_{L+l'}})}{N(E_0+I_{\Lambda_L})}=1.
\end{equation}
Then there exists $\tau=\tau(\rho)$, such that, if, for $\Lambda$ large, one has

\begin{equation}
N(E_0+I_\Lambda)|I_\Lambda|^{-1-\tilde{\rho}}\geq 1\text{    and     }|\Lambda|^{\delta}.N(E_0+I_\Lambda)\underset{\Lambda\to\Z}{\rightarrow}+\infty
\end{equation}

for some $\delta>0$ and $\tilde{\rho}>0$, such that $\dfrac{\delta\tilde{\rho}}{1+\tilde{\rho}}<\tau$.
Then, with probability 1, as $\Lambda\to\infty$, $DLS(x;E_0+I_\Lambda,\omega,\Lambda)$ converge uniformly to the distribution $x\to e^{-x}$, that is, with probability 1,
\begin{equation}
\sup_{x\geq 0} |DLS(x;E_0+I_\Lambda,\omega,\Lambda)-e^{-x}|\underset{\Lambda\to\Z}{\rightarrow}0.
\end{equation}
\end{theo}

Theorem~\ref{ULLS2},  Theorem~\ref{AILP} and Theorem~\ref{LSS} shows that the statistics behave as if the images of eigenvalues by the IDS were i.i.d, uniformly distributed random variables. Other results about spectral statistics can be found in \cite{2010arXiv1011.1832G}.

\section{One-dimensional quantum graphs with random vertex coupling}

In this section we use the previous results to study the spectral statistics for one-dimensional quantum graphs with random vertex coupling. In the study of random Schrödinger operators, some results are only proved for discrete model. For instance, to the best of our knowledge, this is the case of decorrelation estimate for distinct eigenvalues. Thus, it is natural to try to reduce the study of continuous models to the study of discrete models. In the case of quantum graphs with random vertex coupling, this reduction can be done and is very effective. Indeed, for this model, if all but one random variables are fixed, the perturbation is of rank one. Using \cite{KP2008}, the study of the spectrum of this model is reduce to the study of a family of energy-dependent discrete operator with random potential.  

Let $L\in\mathbb{N}$ and  $(\omega_n)$ be non-negative random variables with a common compactly supported bounded density. Consider $H_\omega(\Lambda_L)=\Delta$ on the space 
\begin{equation*}
\bigoplus_{i=0}^L \mathcal{H}^2([n,n+1])
\end{equation*}
 and satisfying the following boundary conditions : 
\begin{align}
&\forall i\in\Z, f_i(1)=f_{i+1}(0)\\
&\forall i\in\Z, f_{i+1}'(0)-f_i'(1)=\omega_i f_i(1)
\end{align}
Then, we know from \cite{exner2005solvable} that this operator is self-adjoint.

Let $H^0=\Delta$ on the space $\bigoplus_{i=0}^L \mathcal{H}^2([n,n+1])$ with Dirichlet boundary conditions at each vertex. $H^0$ is the direct sum of Laplacians on each of the intervals $(n,n+1)$, $n\in\{1,\dots,L\}$, with Dirichlet conditions. So the spectrum of $H^0$ is $\pi^2\mathbb{N}^*$. For $E\notin\pi^2\mathbb{N}^*$ define $M(E)=\dfrac{\sqrt{E}}{\sin(\sqrt{E})}(-\Delta+\cos(\sqrt{E})I_d)$ and $A_\omega:=diag(\omega_n)$.
We know from \cite{KP2008} that

\begin{prop}\label{transf}
If $E\notin\pi^2\mathbb{N}^*$, then $E\in sp(H_\omega)$ iif $0\in sp(M(E)-A_\omega)$ iif $0\in sp(-\Delta+V_\omega(E))$ where $V_\omega(E)(n)=cos(\sqrt{E})- \dfrac{\sin(\sqrt{E})}{\sqrt{E}}\omega_n$.
\end{prop}

As $\parallel V_\omega(E)-V_\omega(E')\parallel\leq C|E-E'|$ we have the following result : 

\begin{prop}
If $I=[E_0-\epsilon,E_0+\epsilon]\cap\pi^2\mathbb{N}^*=\emptyset$ then
\begin{equation}
tr\left(\textbf{1}_I(H_\omega(\Lambda_L)\right)\leq tr\left(\textbf{1}_{[-\epsilon,+\epsilon]}(-\Delta+V_\omega(E_0))_{\Lambda_L}\right).
\end{equation}

\end{prop}
Thus, when the $(\omega_n)_n$ are independent, we can use the Wegner and Minami estimates proved in \cite{CGA09} and we obtain : 

\begin{theo}\label{theo-minami}
	If $I=[E_0-\epsilon,E_0+\epsilon]\cap\pi^2\mathbb{N}^*=\emptyset$ , and $J=[E_0-\epsilon',E_0+\epsilon']\cap\pi^2\mathbb{N}^*=\emptyset$ with $\epsilon\leq \epsilon'$
	\begin{center}
	$\mathbb{P}(tr(\chi_I(H_\omega))\geq k)\leq C\left(\dfrac{\sqrt{E}}{\sin(E_0)}\right)^k\epsilon^k L^k$,\\
	$\mathbb{P}(tr(\chi_I(H_\omega))\geq 1,tr(\chi_J(H_\omega))\geq 2)\leq C\left(\dfrac{\sqrt{E}}{\sin(E_0)}\right)^2\epsilon\epsilon' L^2$.
	\end{center}
\end{theo}
If the $(\omega_n)_n$ are not independent, we still have a Minami estimate as in Theorem~\ref{mina-D} as long as we have a Wegner estimate.
Now we assume that the ($\omega_n)_n$ are independent and prove a decorrelation estimate.
\begin{theo}\label{decvert}
Let $\beta\in(1/2,1)$. For $\alpha\in (0,1) $ and $E,E'\notin\pi^2\mathbb{N}^*$, then for any $c>1$, $L$ sufficiently large and $cL^\alpha\leq l\leq L^\alpha/c$ we have.
\begin{displaymath}
\mathbb{P}\left(
\begin{aligned}
sp(H_\omega(\Lambda_l))\cap[E-L^{-1},E+L^{-1}]\neq\emptyset,\\
sp(H_\omega(\Lambda_l))\cap[E'-L^{-1},E'+L^{-1}]\neq\emptyset
\end{aligned}
\right)\leq C\left(\dfrac{l}{L}\right)^2e^{(\log L)^\beta}.
\end{displaymath}
\end{theo}

 Proposition~\ref{transf} shows that this theorem is equivalent to the following :

\begin{theo}\label{reddecvert}
Let $\beta\in(1/2,1)$. For $\alpha\in (0,1) $ and $E,E'\notin\pi^2\mathbb{N}^*$, then for any $c>1$, $L$ sufficiently large and $cL^\alpha\leq l\leq L^\alpha/c$ we have.
\begin{displaymath}
\mathbb{P}\left(
\begin{aligned}
 sp(-\Delta+V_\omega(E)\restriction{)}{\Lambda_l}\cap[-L^{-1},+L^{-1}]\neq \emptyset,\\
sp(-\Delta+V_\omega(E')\restriction{)}{\Lambda_l}\cap[-L^{-1},+L^{-1}]\neq \emptyset
\end{aligned}
\right)\leq C\left(\dfrac{l}{L}\right)^2e^{(\log L)^\beta}
\end{displaymath}
\end{theo}

\begin{proof}
We remark that for $E\neq E'$ if $\cos(\sqrt{E})=\cos(\sqrt{E'})$ then $\dfrac{sin(\sqrt{E})}{\sqrt{E}}+\dfrac{sin(\sqrt{E'})}{\sqrt{E'}}\neq 0$ and $\dfrac{sin(\sqrt{E})}{\sqrt{E}}-\dfrac{sin(\sqrt{E'})}{\sqrt{E'}}\neq 0$ so Theorem~\ref{reddecvert} is a consequence of Theorem~\ref{deco}. This complete the proof of Theorem~\ref{decvert}.
\end{proof}

As the proof of Theorem~\ref{decvert} suggest, decorrelation estimates for \\multi-dimensional quantum graphs with random vertex coupling should come from a transformation as in Proposition~\ref{transf} and decorrelation estimates as in Theorem~\ref{deco}, but for multi-dimensional discrete Schr\"{o}dinger operators.

 \appendix
 
 \section{Properties of finite-difference equations of order two}
 In this section we prove some results that were used in section 2. They are discrete equivalents of those found in \cite[Appendix]{K10}.
 \begin{lem}\label{progress}
There exists $C>0$ such that for all $n\in\{0,\dots,L\}$
\begin{equation}
\dfrac{1}{C} r_u(n+1)\leq r_u(n)\leq C r_u(n+1).
\end{equation}
\end{lem}

\begin{lem}\label{not0}
There exists $\epsilon_0>0$ such that for all $n\in\{0,\dots,L\}$ and $0\leq\epsilon\leq \epsilon_0$, if $|\sin(\phi_u(n))|\leq \epsilon$ then $|\sin(\phi_u(n+1))|\geq \epsilon_0$.
\end{lem}
\begin{proof}
\normalfont
Let $C>\sup \text{supp}(\omega_n)$ and $\epsilon_0$ such that $ \dfrac{1}{M^2}\sqrt{1-\epsilon_0^2}-CM\epsilon_0>\dfrac{1}{2}$.\\
Let $n\in\{0,\dots,L\}$ and suppose that $|\sin(\phi_u(n))|\leq \epsilon\leq \epsilon_0$ then
as $u(n+1)=\omega_n u(n)-u(n-1)$ we have $\dfrac{u(n+1)}{r_u(n)}=\dfrac{1}{a_{n+1}}\big[\omega_n \sin(\phi_u(n))-a_n\cos(\phi_u(n))\big]$. So
\begin{equation}
\dfrac{|u(n+1)|}{r_u(n)}\geq \dfrac{1}{M^2}|\cos(\phi_u(n))|-CM\epsilon_0\geq \dfrac{1}{M^2}\sqrt{1-\epsilon_0^2}-CM\epsilon_0>\dfrac{1}{2}.
\end{equation}
Then, as $\dfrac{r_u(n)}{r_u(n+1)}\geq \dfrac{1}{C}$ we have $|\sin(\phi_u(n+1))|\geq \dfrac{1}{2C}$. This complete the proof of Lemma~\ref{not0}.
\end{proof}

\begin{lem}\label{coli}
There exists $\epsilon_0>0$ such that for all $n\in\{0,\dots,L\}$ and $0\leq\epsilon\leq \epsilon_0$, if $E<\epsilon$ and if $\delta\phi(n)\in[0,\epsilon]$ then $\delta\phi(n+1)\in \left[0,\frac{\epsilon}{\epsilon_0}\right]$.
\end{lem}
\begin{proof}
\normalfont
As $\delta\phi(n)\in[0,\epsilon]$, we have $ | \sin(\phi_u(n))-\sin(\phi_v(n)) | \leq \epsilon$ and  \\
$ | \cos(\phi_u(n))-\cos(\phi_v(n)) | \leq \epsilon$. Therefore, we have :
\begin{align*}
\dfrac{u(n+1)}{r_u(n)}=&\dfrac{\omega_n}{a_{n+1}} \sin(\phi_u(n))-\dfrac{a_n}{a_{n+1}}\cos(\phi_u(n))\\
					  =&\dfrac{\omega_n}{a_{n+1}} \sin(\phi_v(n))-\dfrac{a_n}{a_{n+1}}\cos(\phi_v(n))+K_\epsilon
					  =\dfrac{v(n+1)}{r_v(n)}+K'_\epsilon
\end{align*}
where
\begin{align*}
|K_\epsilon|&=\dfrac{\big|\omega_n \sin(\phi_v(n))-\omega_n \sin(\phi_u(n))+a_n\cos(\phi_u(n))-a_n\cos(\phi_v(n))\big|}{|a_{n+1}|} \\
&\leq M(C+M)\epsilon\leq C\epsilon.
\end{align*}
As $K'_\epsilon=K_\epsilon+E\sin(\phi_v(n))$, we have $|K'_\epsilon|\leq (C+1)\epsilon$	.\\
So 
\begin{equation*}
\left(\dfrac{u(n+1)}{r_u(n)}\right)^2=\left(\dfrac{v(n+1)}{r_v(n)}\right)^2+C_\epsilon
\end{equation*}
 with $|C_\epsilon|\leq (C+2)\epsilon$.
 Besides 
 \begin{equation*}
 \left(\dfrac{u(n)}{r_u(n)}\right)^2=\sin(\phi_u(n))^2=\left(\dfrac{v(n)}{r_v(n)}\right)^2+ M_\epsilon
 \end{equation*}
 where $|M_\epsilon|\leq 2\epsilon$. Therefore, 
 \begin{align*}
\left(\dfrac{r_u(n+1)}{r_u(n)}\right)^2&=\left(\dfrac{u(n+1)}{r_u(n)}\right)^2+\left(\dfrac{u(n)}{r_u(n)}\right)^2\\
&=\left(\dfrac{v(n+1)}{r_v(n)}\right)^2+\left(\dfrac{v(n)}{r_u(n)}\right)^2+C'_\epsilon + M_\epsilon \\ 
&=\left(\dfrac{r_v(n+1)}{r_v(n)}\right)^2+M'_\epsilon.
\end{align*}
So as we have $\dfrac{1}{C}\leq \dfrac{r_v(n+1)}{r_v(n)}\leq C$ and the same for u, for some $\tilde{C}$ we have
\begin{displaymath}
\left|1-\left(\dfrac{r_v(n+1)}{r_v(n)}\right)\left(\dfrac{r_u(n)}{r_u(n+1)}\right)\right|\leq \left|1-\left(\dfrac{r_v(n+1)}{r_v(n)}\right)^2\left(\dfrac{r_u(n)}{r_u(n+1)}\right)^2\right|
\leq \tilde{C} \epsilon.
\end{displaymath}
Now 
\begin{equation}
\dfrac{u(n+1)}{r_u(n+1)}=\dfrac{v(n+1)}{r_v(n+1)}\left(\dfrac{r_v(n+1)}{r_v(n)}\right)\left(\dfrac{r_u(n)}{r_u(n+1)}\right)+\left(\dfrac{r_u(n)}{r_u(n+1)}\right)K'_\epsilon
\end{equation}
So 
\begin{equation}
\begin{split}
\Big|\sin(\phi_u(n+1))-\sin(\phi_v(n+1))\Big|\leq\left|1-\left(\dfrac{r_v(n+1)}{r_v(n)}\right)\left(\dfrac{r_u(n)}{r_u(n+1)}\right)\right| \\
+\left(\dfrac{r_u(n)}{r_u(n+1)}\right)|K'_\epsilon| 
\end{split}
\end{equation}
and for some $N>0$
\begin{equation}
\Big|\sin(\phi_u(n+1))-\sin(\phi_v(n+1))\Big|\leq N\epsilon
\end{equation}
Besides, $r_u(n+1)\cos(\phi_u(n+1))=r_u(n)\sin(\phi_u(n))$ so a similar proof shows that
\begin{equation}
\Big|\cos(\phi_u(n+1))-\cos(\phi_v(n+1))\Big|\leq N\epsilon.
\end{equation}
Now take $\epsilon_0\leq \dfrac{\pi}{4N}$, then $\delta\phi(n+1)\in[0,2N\epsilon]$.
\end{proof}

In the same way one proves
\begin{lem}\label{oppo}
There exists $\epsilon_0>0$ such that for all $n\in\{0,\dots,L\}$ and $0\leq\epsilon\leq \epsilon_0$, if $E<\epsilon$ and if $\delta\phi(n)\in[\pi-\epsilon,\pi]$ then $\delta\phi(n+1)\in \left[\pi-\dfrac{\epsilon}{\epsilon_0},\pi\right]$.
\end{lem}

\section{Proof of Theorem~\ref{example}}\label{ap2}
Let $(\omega_i)_{i\in\N}$ be independent random variables with a common, compactly supported density satisfying
\begin{equation}\label{KlAssump}
\exists \rho_0>0, \exists \epsilon_0>0, \forall \epsilon\in[0,\epsilon], \int_\R \sup_{u\in[-1,1]} |g(t+\epsilon u)-g(t)| dt \leq \left(\frac{\epsilon}{\epsilon_0}\right)^{\rho_0}.
\end{equation}
If $(u_i)_{i\in\N}$ are sequences that are not all zero and 
\begin{equation}
V_\omega(n):=\sum_{i\in\R} \omega_i u_i(n-i)
\end{equation}
then define $H_\omega:=H_0+V_\omega$ on $\ell^2(\Z)$ where $H_0$ is a lower semi-bounded, self-adjoint perturbation of the laplacian. Then, the spectrum is almost surely constant and, as in \cite{Kl95}, we have a Wegner estimate at the lower edge of the spectrum. Indeed, let $E_{\inf}$ be the infimum  of the almost sure spectrum. 

\begin{theo}\label{KlW}
There exists $E_0>E_{\inf}$, $C>0$, $q_0>0$ such that , for any $l\geq 2$, any $\epsilon>0$ and any $E\in[E_{\inf},E_0]$,
\begin{equation}
\mathbb{P}\left(d(E,\sigma(H_\omega(\Lambda_l)))<\epsilon\right)\leq C \epsilon^{\rho_0}l^{q_0}.
\end{equation}

\end{theo}

Now taking $u_{2i+1}:=0$ and $u_{2i}(0)=1,u_{2i}(-1)=-1$ and zero elsewhere, we get 
$V_\omega(2i)=-V_\omega(2i+1)=\omega_{2i}$. Suppose now that $\omega_1$ has support equal to $[0,1]$ . Then, we know that almost surely, $\sigma(H_\omega(\Lambda_l))\subset [-3,3]$.

 Now, for $\phi\in\ell^2(\Z)$ define
$\Phi(n)=\begin{pmatrix}
\phi(n)\\\phi(n-1)
\end{pmatrix}$. If $\phi$ is an eigenvector for the energy $E$ then the one-step transfer matrix going from $n$ to $n+1$ is $\begin{pmatrix}
V_\omega(n)-E & -1\\
1 & 0
\end{pmatrix}.$ Thus, the two-step transfer matrix going from $2n$ to $2n+2$ is equal to 
\begin{align}
-T_n:&=\begin{pmatrix}
-(\omega_n-E)(\omega_n+E)-1 & -(\omega_n-E)\\
-(\omega_n+E) & -1
\end{pmatrix}
\\&=
-\begin{pmatrix}
(\omega_n-E)(\omega_n+E)+1 & (\omega_n-E)\\
(\omega_n+E) & 1
\end{pmatrix}.
\end{align}
From now on we study the matrices $(T_n)_n$ and not the two-step matrices directly, but what we prove for the $(T_n)_n$ also holds for $(-T_n)_n$.
The matrices $(T_n)_{n\in\Z}$ are i.i.d and following \cite{bougerol1985products} we prove that G, the closed subgroup of $SL_2(\R)$ generated by the support of the matrices is not compact and that the orbit of each direction in $\mathbb{P}(\R^2)$, the projective plan, has at least three elements. 

Fix $E\in[-3,0]$, the case $E\in[0,3]$ being handled in the same way. As $\omega_n$ has support equal to $[0,1]$, the  following matrix belongs to G, taking $\omega_1=0$ :
\begin{equation}
A:=
\begin{pmatrix}
1-E^2&-E\\E&1
\end{pmatrix}.
\end{equation}
We have $\text{tr}(A)=2-E^2$, so if $E\in[-3,-2)$, $\text{tr}(A)<-2$ and A is hyperbolic so G is not bounded, thus, not compact.

Now take $\tilde{x}$ the class of $x$ in $\mathbb{P}(\R^2)$ and suppose $E\in[-3,-2)$, the case $E\in(2,3]$ being handled in the same way. The matrix  A has one eigenvalue in $(1,+\infty)$ associated to the eigenvector $e_1$ and an other in $(0,1]$ associated to $e_2$. Thus, any orbit of a direction different of $e_1$ and $e_2$ has an infinity of elements. But as $\omega_n$ takes all values in $(0,1)$, for $\delta\in(0,-2-E)\subset(0,1)$ and $\omega_n=\delta$, G contains an other hyperbolic matrix $C$  with $\text{tr}(C)=2-E^2+\delta<2-(-2-\delta)^2+\delta=-2-\delta-\delta^2<-2$ but with other eigenvectors and eigenvalues. So the orbits of $e_1$ and $e_2$ also contain an infinity of elements.

Now, fix $E\in(-2,-1)$. Then, taking $\omega_1=1$, the following matrix belongs to G :

\begin{equation}
B:=
\begin{pmatrix}
2-E^2&1-E\\1+E&1
\end{pmatrix}.
\end{equation}
Thus,
\begin{equation}
\text{tr}(B)=3-E^2
\end{equation}
and both $A$ and $B$ are elliptic matrices. Now, we compute
\begin{equation}
\text{tr}(A^2)=(1-E^2)^2-2E^2+1=E^4-4E^2+2\in(-2,2)
\end{equation}
Thus, $A$ and $A^2$ are elliptic matrices. Now, we know that an elliptic matrix has no fixed point in $\mathbb{P}(\R^2)$. Therefore, the orbit of any point $\tilde{x}\in\mathbb{P}(\R^2)$ contains the set $\{\tilde{x},A\tilde{x},A^2\tilde{x}\}$ which have three elements.
Now, a computation shows that 
\begin{equation}
AB-BA=\left(\begin{array}{cc}
-2\cdot E & -E^{2}+E \\
E^{2}+E & 2\cdot E
\end{array}\right).
\end{equation}
Thus, the matrices $A$ and $B$ do not commute. Thus, their commutator \\
$A^{-1}B^{-1}AB$ is an hyperbolic matrix that belongs to $G$ and $G$ is not compact.

Now, suppose $E=-2$. Then, $A$ is a parabolic matrix and G is not compact. We compute $\text{tr}(B)=-1$ and $\text{tr}(B^2)=(2-E^2)^2+2(1-E)(1+E)+1=E^4-6E^2+9=1$. Thus, $B$ and $B^2$ are elliptic matrix and the orbit of any element $\tilde{x}\in\mathbb{P}(\R^2)$ has at least three elements.

Now, suppose $E[-1,0)$ and take $\delta\in(0,-E)$. Then, $A$ and $A^2$ are still elliptic matrices, so the orbit of any direction has at least three elements.
As $\omega_1$ is uniformly distributed on $[0,1]$, the matrix 
\begin{equation}
C_\delta:=
\begin{pmatrix}
1+\delta^2-E^2&1-E+\delta\\1+E+\delta&1
\end{pmatrix}
\end{equation}
belongs to $G$. We compute $\text{tr}(C)=2+\delta^2-E^2\in(-2,2)$, so $C$ is an elliptic matrix that do not commute with $A$. Hence, their commutator is a hyperbolic matrix and the group G is not compact.

Eventually, suppose $E=0$. Then, $B$ is an hyperbolic matrix, so the group G is not compact. Now, the matrix $C_\delta$ for some $\delta\in(0,1)$ is also hyperbolic, but with different eigenvalues and eigenvectors. Thus, the orbits of any element has at least three elements.

This complete the proof of 
\begin{prop}
If the support of $\omega_1$ is the interval $[0,1]$  and $E\in[-3,-3]$,  the group G is not compact and the orbit of each direction has at least three elements. 
\end{prop}

We can now apply the Furstenberg theorem and as in \cite[Sections B-II-4 and B-II-6]{bougerol1985products}, using the regularity of the Lyapunov exponent, we obtain that the density of state is H\"older continuous in $[-3,3]$. Combining this with Theorem~\ref{KlW}, we have both results at the lower edge of the spectrum and this prove Theorem~\ref{example}.

\bibliographystyle{plain}
\bibliography{biblio}
\end{document}